\documentclass[oneside,english]{amsart}
\usepackage[T1]{fontenc}
\usepackage[latin9]{inputenc}
\usepackage{babel}
\usepackage{amstext}
\usepackage{amsthm}
\usepackage{amssymb}
\usepackage{graphicx}
\usepackage{setspace}
\usepackage[unicode=true,pdfusetitle,
 bookmarks=true,bookmarksnumbered=false,bookmarksopen=false,
 breaklinks=false,pdfborder={0 0 1},backref=false,colorlinks=false]
 {hyperref}

\makeatletter
\numberwithin{equation}{section}
\numberwithin{figure}{section}
\theoremstyle{remark}
\newtheorem*{rem*}{\protect\remarkname}
\theoremstyle{plain}
\newtheorem{thm}{\protect\theoremname}
\theoremstyle{definition}
\newtheorem{defn}[thm]{\protect\definitionname}
\theoremstyle{plain}
\newtheorem{prop}[thm]{\protect\propositionname}
\theoremstyle{plain}
\newtheorem{cor}[thm]{\protect\corollaryname}
\theoremstyle{remark}
\newtheorem{rem}[thm]{\protect\remarkname}
\theoremstyle{definition}
\newtheorem{problem}[thm]{\protect\problemname}
\theoremstyle{definition}
\newtheorem{example}[thm]{\protect\examplename}
\theoremstyle{plain}
\newtheorem{lem}[thm]{\protect\lemmaname}

\@ifundefined{showcaptionsetup}{}{%
 \PassOptionsToPackage{caption=false}{subfig}}
\usepackage{subfig}
\makeatother

\providecommand{\corollaryname}{Corollary}
\providecommand{\definitionname}{Definition}
\providecommand{\examplename}{Example}
\providecommand{\lemmaname}{Lemma}
\providecommand{\problemname}{Problem}
\providecommand{\propositionname}{Proposition}
\providecommand{\remarkname}{Remark}
\providecommand{\theoremname}{Theorem}

\begin{document}
\title[Singularities of symmetric symbols on surfaces]{Singularities of the eigenvalue functions for first order symmetric
symbols on rank two vector bundles over surfaces}
\author{Carlos Valero Valdes\\
Departamento de Matematicas \\
Universidad de Guanajuato\\
Guanajuato, México}
\date{01 Oct 2019}
\begin{abstract}
For a rank two bundle $F$ over a surface $X$, we study the set of
singularities $\mathcal{\mathcal{M}}_{\sigma}$ of the eigenvalue
functions of symmetric symbols $\sigma$ associated to first order
differential operators on $F$. We prove that the existence of these
singularities follows from topological considerations. We define $\deg(\mathcal{M}_{\sigma})$
(the degree of $\mathcal{\mathcal{M}}_{\sigma}$) and show that it
can be used to ``count'' the number of directions at which special
optical phenomena occur. For the case when $F=TX$ we compute a formula
for $\deg(\mathcal{\mathcal{M}}_{\sigma})$ in terms of the Euler
characteristics of $X$ and $\mathcal{N}_{\sigma}$, where $\mathcal{N}_{\sigma}\subset X$
is a manifold with boundary associated to $\sigma$. 
\end{abstract}

\maketitle
\global\long\def\CC{\mathbb{C}}%

\global\long\def\RR{\mathbb{R}}%

\global\long\def\map{\rightarrow}%

\global\long\def\mult{\mathcal{M}}%

\global\long\def\EE{\mathcal{E}}%

\global\long\def\OO{\mathcal{O}}%

\global\long\def\FH{\mathcal{F}}%

\global\long\def\CH{\mathcal{H}}%

\global\long\def\tangent{T}%

\global\long\def\cotangent{\tangent^{*}}%

\global\long\def\conormal{\mathcal{C}}%

\global\long\def\CS{\mathcal{C}}%

\global\long\def\SS{\mathcal{S}}%

\global\long\def\KK{\mathcal{K}}%

\global\long\def\NN{\mathcal{N}}%

\global\long\def\sym#1{\hbox{S}^{2}#1}%

\global\long\def\symz#1{\hbox{S}_{0}^{2}#1}%

\global\long\def\proj#1{\hbox{P}#1}%

\global\long\def\SO{\hbox{SO}}%

\global\long\def\GL{\hbox{GL}}%

\global\long\def\U{\hbox{U}}%

\global\long\def\tr{\hbox{tr}}%

\global\long\def\ZZ{\mathbb{Z}}%

\global\long\def\der#1#2{\frac{\partial#1}{\partial#2}}%

\global\long\def\covder{\hbox{D}}%

\global\long\def\diff{d}%

\global\long\def\dero#1{\frac{\partial}{\partial#1}}%

\global\long\def\END{\hbox{End}}%

\global\long\def\LL{\mathcal{L}}%

\global\long\def\ind{\hbox{ind}}%

\global\long\def\and{\hbox{\,\,and\,\,}}%

\section{Introduction}

\begin{figure}[h]

\subfloat[Upon entering a bi-axial crystal along one of its optical axes, a
light beam refracts as a slanted cone and then emerges as a cylinder.]{

\includegraphics[scale=0.35]{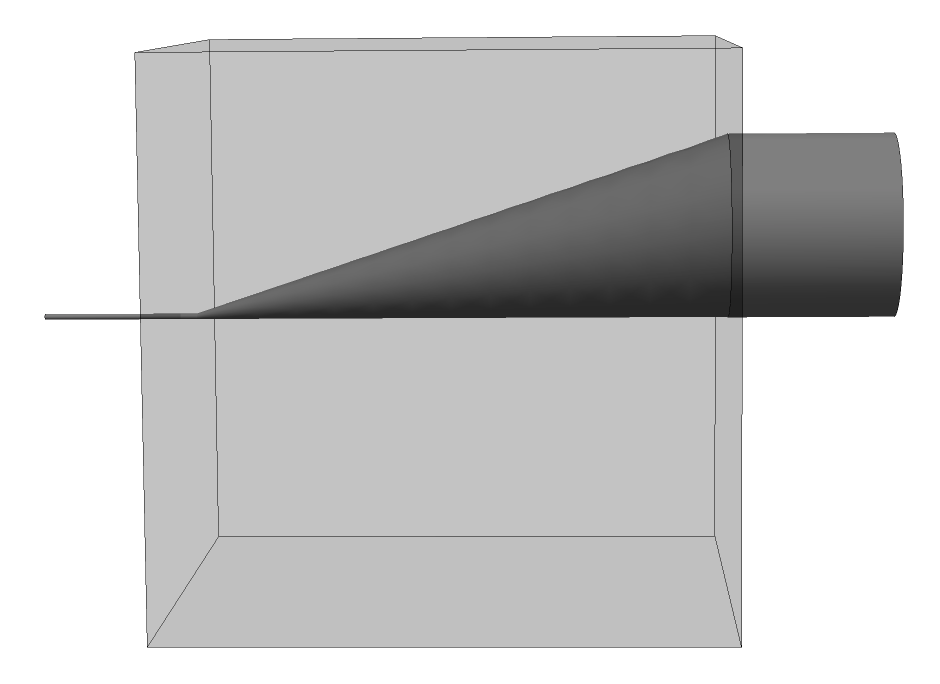}

}\subfloat[\label{fig:Fresnel-surface}Fresnel surface associated to a bi-axial
crystal. This surface has four conical singularities, which account
for the phenomenon of conical refraction.]{\includegraphics[scale=0.25]{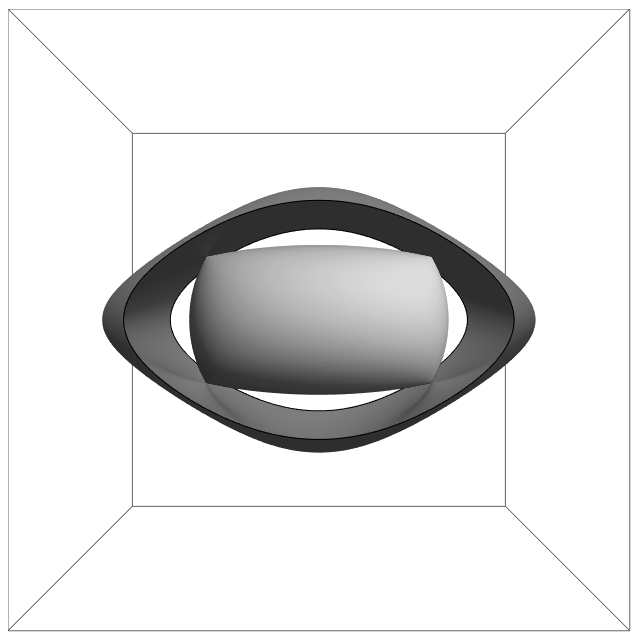}

}\caption{\emph{\label{fig:Conical-refraction}} Conical refraction}

\end{figure}
On the 22nd of October 1832, in the ``Third Supplement to an Essay
on the Theory of Systems of Rays'', William R. Hamilton predicted
the phenomenon of \emph{conical refraction} for bi-axial crystals
based on the existence of conical singularities in its \emph{Fresnel
surface} (see Figure \ref{fig:Conical-refraction}). The existence
of this phenomenon has been applied to diverse fields, like optical
trapping, free-space optical communications, polarization metrology,
super-resolution imaging, two-photon polymerization, and lasers (see
\cite{kn:CRApplications}).

The propagation of electromagnetic waves in a homogeneous and an-isotropic
medium can be modeled by Maxwell's equations (see \cite[pg. 678]{kn:born})
\begin{align}
\der Et-\epsilon^{-1}(\nabla\times B) & =0,\label{eq:MaxwellsA}\\
\der Bt+\nabla\times E & =0,\label{eq:MaxwellsB}
\end{align}
and
\begin{equation}
\nabla\cdot(\epsilon E)=\nabla\cdot B=0.\label{eq:MaxwellC}
\end{equation}
In the above formulas we have assumed that the speed of light is equal
to $1$ and that $\epsilon$ is a $3\times3$ positive definite symmetric
matrix (known as the \emph{dielectric tensor}). If for $E_{0},B_{0}\in\RR^{3}$
and $(\xi,\tau)\in(\RR^{3}\times\RR)^{*}$ we consider fields that
move along planar waves, i.e
\begin{align*}
E(x,t) & =E_{0}\exp(i(\xi\cdot x-\tau t)),\\
B(x,t) & =B_{0}\exp(i(\xi\cdot x-\tau t)),
\end{align*}
by substitution into equations \ref{eq:MaxwellsA} and \ref{eq:MaxwellsB}
we obtain
\begin{equation}
(\tau I-\sigma_{\epsilon}(\xi))\left(\begin{array}{c}
E_{0}\\
B_{0}
\end{array}\right)=0,\label{eq:MaxwellKernel}
\end{equation}
where $I$ is the $6\times6$ identity matrix, 
\[
\sigma_{\epsilon}(\xi)=\left(\begin{array}{cc}
0 & -\epsilon^{-1}P(\xi)\\
P(\xi) & 0
\end{array}\right)
\]
and
\[
P(\xi)=\left(\begin{array}{ccc}
0 & -\xi_{3} & \xi_{2}\\
\xi_{3} & 0 & -\xi_{1}\\
-\xi_{2} & \xi_{1} & 0
\end{array}\right).
\]
To find non-trivial solutions of Equation \ref{eq:MaxwellKernel}
we need to solve the\emph{ characteristic equation}
\begin{equation}
\det(\tau I-\sigma_{\epsilon}(\xi))=0.\label{eq:CharacteristicEquationMaxwell}
\end{equation}
The matrix $\sigma_{\epsilon}(\xi)$ has $0$ as an eigenvalue of
multiplicity $2$ and the other four eigenvalues are of the form $-\lambda_{\epsilon,1}(\xi),\lambda_{\epsilon,1}(\xi),-\lambda_{\epsilon,2}(\xi)$
and $\lambda_{\epsilon,2}(\xi)$, where the functions $\lambda_{\epsilon,1}$
and $\lambda_{\epsilon,2}$ are homogenous of degree one. The Fresnel
surface associated to $\epsilon$ is the set
\[
\FH_{\epsilon}=\FH_{\epsilon}^{1}\cup\FH_{\epsilon}^{2}
\]
where 
\[
\FH_{\epsilon}^{i}=\{\xi\in(\RR^{3})^{*}|\lambda_{\epsilon,i}(\text{\ensuremath{\xi)=1}}\}.
\]
The sets $\FH_{\text{\ensuremath{\epsilon}}}^{1}$ and $\FH_{\epsilon}^{2}$
correspond to the inner and outer sheets of the surface shown in Figure
\ref{fig:Fresnel-surface}. The existence of the four conical singularities
in $\FH_{\epsilon}$ follows from the fact that $\lambda_{\epsilon,1}$
and $\lambda_{\epsilon,2}$, when restricted to the sphere $S^{2}=\{\xi\in(\RR^{3})^{*}|<\xi,\xi>=1\}$,
are non-smooth at the \emph{multiplicity set}
\[
\mult_{\epsilon}=\{\xi\in S^{2}|\lambda_{\epsilon,1}(\xi)=\lambda_{\epsilon,2}(\xi)\}.
\]

\subsection{First order hyperbolic differential operators with constant coefficients}

We can study the problem discussed above for Maxwell's equations in
a more general setting. Consider the first order linear partial differential
equation 
\begin{equation}
H_{L}u=0\hbox{\,\,\,for\,\,\,}u:\RR^{n}\times\RR\rightarrow\RR^{k},\label{eq:LinearPDEOrder1ConstantCoeffs}
\end{equation}
where 
\begin{equation}
H_{L}=\dero tI+L\hbox{\,\,for\,\,}L=\sum_{i=1}^{n}A_{i}\dero{x_{i}}\hbox{\,\,for\,\,}A_{i}\in M_{k\times k}(\RR),\label{eq:GeneralFirstOrderHyperbolicOperator}
\end{equation}
and $I$ is the identity $k\times k$ matrix.
\begin{rem*}
Maxwell's equations \ref{eq:MaxwellsA} and \ref{eq:MaxwellsB} can
be written in the form \ref{eq:LinearPDEOrder1ConstantCoeffs} for
an operator $L$ with appropriate coefficients $A_{1},A_{2},A_{3}\in M_{6\times6}(\RR)$. 

If for $u_{0}$ in $\RR^{k}$ we substitute fields of the form $u=u_{0}\exp(\xi\cdot x-\tau t)$
into \ref{eq:LinearPDEOrder1ConstantCoeffs} we obtain the equation
\begin{equation}
(\tau I-\sigma_{L}(\xi))u_{0}=0\label{eq:PolarizationEqConstCoeffs}
\end{equation}
where 
\[
\sigma_{L}(\xi_{1},\ldots,\xi_{n})=\sum_{i=1}^{n}A_{i}\xi_{i}\hbox{\,\,\,for\,\,\,}\xi=(\xi_{1},\ldots,\xi_{n}).
\]
 To solve \ref{eq:PolarizationEqConstCoeffs} we need to solve the
characteristic equation $\det(\sigma(\xi)-\tau I)=0$. As in the case
of Maxwells equations, we are interested in the points $\xi\in S^{n-1}$
at which $\sigma_{L}(\xi)$ has multiple eigenvalues. This condition
has been studied by P. Lax in \cite{kn:lax}, F. John in \cite{kn:john}
and B.A. Khesin in \cite{kn:khesin}. In \cite{kn:FriedlandCrossing}
Friedland, Robbin and Sylvester elegantly reduced this problem to
Adams\textquoteright{} celebrated theorem on linearly independent
vector fields on the sphere. The problem is also related to the 1929
celebrated von Neumann\textendash Wigner theorem on the non-crossing
rule (see \cite{kn:NeumannCrossing}). 
\end{rem*}

\subsection{First order hyperbolic differential operators on vector bundles }

The above ideas can be generalized even further to the case of first
order hyperbolic linear partial differential equations on bundles.
More concretely, consider a Riemannian manifold $X$ with metric $<,>.$
The cotangent bundle $\cotangent X$ inherits a metric from that of
$\tangent X$ by letting the map $v\mapsto<v,\cdot>$ be an isometry
. Consider a vector bundle $F$ over $X$, a linear differential operator
$L:C^{\infty}(F)\rightarrow C^{\infty}(F)$ of order one, and the
equation
\begin{equation}
H_{L}u=0\hbox{\,\,where\,\,}H_{L}=\frac{\partial}{\partial t}+L\label{eq:HyperbolicEqManifolds}
\end{equation}
where $u:X\times\RR\rightarrow F$ is a time dependent section of
$F$. We are interested in high frequency solutions of \ref{eq:HyperbolicEqManifolds}.
More concretely, for a section $u_{0}:X\rightarrow F$ and a phase
function $\varphi:X\times\RR\rightarrow\RR$ we consider oscillatory
sections of the form
\[
u_{s}(x,t)=u_{0}e^{is\varphi(x,t)}.
\]
By substitution into equation \ref{eq:HyperbolicEqManifolds} and
letting $s\mapsto\infty$ we obtain (for the details of the derivation
this formula the reader is referred to books on asymptotic methods
like \cite{kn:maslov,kn:sternin,kn:guillemin} )
\begin{equation}
\left(\der{\varphi}t(x,t)I+\sigma(d_{x}\varphi(x,t))\right)u_{0}=0,\label{eq:PolarizationOnManifolds}
\end{equation}
where $I:F\rightarrow F$ is the identity morphism and $\sigma:\cotangent X\rightarrow\hbox{End}(F)$
is the \emph{principal symbol} of $L.$ Since we are assuming that
$L$ is a differential operator of degree one, for every $x\in X$
we have that $\sigma$ is a linear map from $\cotangent_{x}X$ to
$\hbox{End}(F_{x})$. A necessary condition for \ref{eq:PolarizationOnManifolds}
to hold is that $\varphi$ satisfies the non-linear partial differential
equation of first order, know as the characteristic equation, given
by
\begin{equation}
\det\left(\der{\varphi}t(x,t)I+\sigma(d_{x}\varphi(x,t))\right)=0.\label{eq:CharacteristicEquationManifolds}
\end{equation}
This equation can be interpreted as follows. For a particle $x:\RR\rightarrow X$
traveling on a wave front defined by the equation
\[
\varphi(x(t),t)=\hbox{constant}
\]
 we have that
\[
\der{\varphi}t(x,t)=-<\nabla_{x}\varphi(x,t),\dot{x}>,
\]
so formula \ref{eq:CharacteristicEquationManifolds} becomes
\[
\det\left(\left\langle n(x,t),\dot{x}\right\rangle I-\sigma\left(n^{\flat}(x,t)\right)\right)=0
\]
for 
\[
n(x,t)=\frac{\nabla_{x}\varphi(x,t)}{|\nabla_{x}\varphi(x,t)|}\hbox{\,\,and\,\,}n^{\flat}(x,t)=\frac{d_{x}\varphi(x,t)}{|d_{x}\varphi(x,t)|}.
\]
We conclude that if a particle moves in a wave front 
\[
Y_{t}=\{x\in X|\varphi(x,t)=\hbox{constant}\}
\]
then the component of its velocity normal to the wave front must be
an eigenvalue of $\sigma(n^{\flat}(x,t))$. Hence, for a unitary co-vector
$\xi\in T_{x}^{*}X$ the eigenvalues of $\sigma(\xi)$ are the allowable
normal velocities of particles moving along wave fronts whose tangent
spaces at $x$ is the kernel of $\xi$. If $\lambda_{\sigma,i}:\cotangent X\rightarrow\RR$
are the eigenvalues of $\sigma$ (ordered as $\lambda_{\sigma,1}\leq\text{\ensuremath{\cdots}}\leq\lambda_{\sigma,k}$)
then we are interested in the non-zero covectors $\xi\in\cotangent X$
at which $\text{\ensuremath{\lambda_{\sigma,i}(\xi)=\lambda_{\sigma,j}(\xi)}}$
for some $i<j$.

\subsection{Symmetric symbols and characteristic functions}

In this paper we will only be interested in the case where the symbol
$\sigma$ is a \emph{symmetric symbol}, i.e we will assume that for
a metric $<,>_{F}$ in $F$ we have that for all $x\in X$ and $\xi\in\cotangent_{x}X$
the following holds
\[
<\sigma(\xi)v,w>_{F}=<v,\sigma(\xi)w>_{F}\hbox{\,\,for\,all\,}v,w\in F_{x}.
\]
If we denote the bundle of symmetric endomorphisms of $F$ by $\sym F$,
then $\sigma$ is a morphism of the form $\sigma:\cotangent X\rightarrow\sym F$.
Symmetric symbols appear in diverse physical problems derived from
variational principles (see \cite[Ch. 8]{kn:arnold}). If $\lambda_{\sigma,1}(\xi),\ldots,\lambda_{\sigma,k}(\xi)$
are the eigenvalues of $\sigma(\xi)$ then the characteristic equation
\ref{eq:CharacteristicEquationManifolds} is equivalent to the \emph{Hamilton-Jacobi
equations}

\begin{equation}
\frac{\partial\varphi}{\partial t}(x,t)+\lambda_{\sigma,i}(d_{x}\varphi(x,t))=0\hbox{\,for\,\,}i=1,\ldots,k,\label{eq:Characteristic Equation Scalar}
\end{equation}
where $k$ is the rank of $F$. We will refer to the functions $\lambda_{\sigma,i}:\cotangent X\rightarrow\RR$
as the eigenvalue functions of the symbol $\sigma$. For hyperbolic
partial differential equations which correspond to generic non-homogeneous
media the existence of non-smooth points of the eigenvalue functions
account for a phenomena known as wave transformation, in which longitudinal
waves convert into transversal waves or vice-versa (see \cite[Ch. 8]{kn:arnold}
and \cite{kn:braam}). Hence, it is important to find methods for
detecting points at which the eigenvalue functions of a symbol become
non-smooth.

\subsection{Plan of the paper}

For the rest of the paper we will assume that $X$ is a Riemannian
compact, connected and oriented surface, and $F$ an $\SO(2)$ vector
bundle of rank $2$ over $X$. For a first order differential operator
acting on $F$ and having a symmetric symbol $\sigma:\cotangent X\rightarrow\sym F$,
we are interested in studying the multiplicity set \emph{
\[
\mult_{\sigma}=\{\xi\in S(\cotangent X)|\lambda_{\sigma,1}(\xi)=\lambda_{\sigma,2}(\xi)\}.
\]
}The reason for restricting our study of $\mult_{\sigma}$ to $X$
and $F$ as above is that in this case we can prove many interesting
properties of $\mult_{\sigma}$, which we hope will pave the way and
motivate the relevant questions for the more general cases. Our main
results can be summarized as follows:
\begin{itemize}
\item We will show that for symbols $\sigma$ in general position the multiplicity
set $\mult_{\sigma}$ is a $1$-dimensional sub-manifold of $S(\cotangent X)$
which can be characterized as the set of points at which the eigenvalue
functions $\lambda_{\sigma,1}$ and $\lambda_{\sigma,2}$ are non-smooth
(see Proposition \ref{pro:singularitiesequationforfresnel}).
\item We will give a topological obstruction for $\mult_{\sigma}$ to be
empty. More concretely, we will show that if $e(TX)$ and $e(F)$
are the Euler classes of $TX$ and $F$, then $e(\tangent X)-2e(F)\not=0$
in $H^{2}(X,\RR)$ implies that $\mult_{\sigma}\not=\emptyset$ (Theorem
\ref{pro:SingCharFunctions}). In particular, if $F=\tangent X$ and
the genus of $X$ is not $0$ then $\mult_{\sigma}$ must always be
non-empty ( \ref{pro:SingCharFunctions}).
\item We will show that for symbols $\sigma$ in general position the set
$\mult_{\sigma}$ is the intersection of $S(\cotangent X)$ with a
real line bundle $\KK_{\sigma}\subset\cotangent X$ over one dimensional
sub-manifold of $\SS_{\sigma}$ of $X$ (Theorem \ref{thm:MultiplicityAsKernelBundle}
and Corollary \ref{cor:SingCharFunctions}). This will allow us to
define the integer $\deg(\mult_{\sigma})\in\ZZ$ as the number of
times that $\KK_{\sigma}$ turns as we move along $\SS_{\sigma}$.
We will provide a physical interpretation of this number to motivate
its relevance.
\item We will prove that when $F=\tangent X$ a symbol $\sigma$ has two
associated sections $v:X\rightarrow F$ and $w:X\rightarrow F\otimes_{\CC}F\otimes_{\CC}F$
(Theorem \ref{thm:ISOMORPHISM} and Corollary \ref{cor:SymbolsFromComplex})
such that
\[
\SS_{\sigma}=\{x\in X||w(x)/v(x)|=1\}
\]
and
\[
\mult_{\sigma}=\bigcup_{x\in\SS_{\sigma}}\left\{ z^{\flat}\in\cotangent_{x}X|z\in\tangent_{x}X\hbox{\,\,and\,\,}z^{2}=\frac{w(x)}{v(x)}\right\} .
\]
From this result we will obtain the following formula (see Theorem
\ref{thm:MainIndexTheorem})
\[
\deg(\mult_{\sigma})=3\chi(X)-2\chi(\NN_{\sigma}),
\]
where 
\[
\NN_{\sigma}=\{x\in X||w(x)/v(x)|\leq1\},
\]
and $\chi(X)$ and $\chi(\NN_{\sigma})$ are the Euler characteristics
of $X$ and $\NN_{\sigma}$, respectively.
\item We conclude the paper with the construction of some explicit examples
of symbols on the plane (the only non-compact case considered) and
the sphere, and use the formula in Theorem \ref{thm:MainIndexTheorem}
to compute $\deg(\mult_{\sigma})$ for these examples.
\end{itemize}

\section{Symbols and Their Multiplicities}

\subsection{Preliminaries}

From now on and for rest of the paper we will assume that $F$ is
Riemannian orientable real vector bundle of rank $2$, over an oriented
Riemannian surface $X$ that is compact and connected. To distinguish
the metrics in $TX$ and $F$ we will denote them by $<,>$ and $<,>_{F}$,
respectively. The metric in $TX$ induces isomorphism maps $\flat:\tangent X\rightarrow\cotangent X$
and $\sharp:\cotangent X\rightarrow TX$, such that for $v,w\in\tangent_{x}X$
and $\xi\in\cotangent_{x}X$ we have that
\[
v^{\flat}(w)=<v,w>\hbox{\,\,and\,\,\,}\xi(w)=<\xi^{\sharp},w>.
\]
The metric in $\tangent X$ induces a metric $\cotangent X$ by letting
\[
<\xi,\eta>=<\xi^{\sharp},\eta^{\sharp}>.
\]
Observe that we are using the same notation for the metrics in $\cotangent X$
and $\tangent X$. 

Let $E_{1}$ and $E_{2}$ be vector bundles over $X$ with projection
maps $\pi_{1}$ and $\pi_{2}$. A section $s:X\rightarrow\hbox{Hom}(E_{1},E_{2})$
can be seen as a bundle morphism $\sigma_{s}:E_{1}\rightarrow E_{2}$
given by
\[
\sigma_{s}(v)=s(x)v\hbox{\,\,for\,\,}x=\pi_{1}(v).
\]
Similarly, a bundle morphism $\sigma:E_{1}\rightarrow E_{2}$ can
be seen as a sections $s_{\sigma}:X\rightarrow\hbox{Hom}(E_{1},E_{2})$
given by
\[
s_{\sigma}(x)v=\sigma(v)\hbox{\,\,for\,\,}v\in\pi_{1}^{-1}(x).
\]

\subsection{Symmetric symbols}

Motivated by the discussion in the introduction, we introduce our
basic objects of study. 
\begin{defn}
A\emph{ symmetric symbol} on the bundle $F\rightarrow X$ is a smooth
bundle morphism $\sigma:\cotangent X\rightarrow\sym F$.
\end{defn}

From now on we will refer to a symmetric symbol simply as a symbol;
since we will not be working with non-symmetric ones. Given that a
symbol $\sigma$ is a morphism of vector bundles, it is completely
determined by the values it take on the sphere bundle
\[
S(\cotangent X)=\{\xi\in\cotangent X|<\xi,\xi>=1\}.
\]
The bundle $\sym F$ has a natural Riemannian metric given by
\[
g(A,B)=\frac{1}{2}\tr(A\circ B)\hbox{\,\,for\,all\,\,}x\in X\hbox{\,and\,}A,B\in\sym{F_{x}}.
\]
This metric induces a norm function on $\sym F$; which for an element
$A\in\sym F$ we will denote by $||A||$. If we let
\[
\sigma_{0}=\sigma-\frac{1}{2}\tr(\sigma)I
\]
be the traceless part of $\sigma$,  the eigenvalue functions $\lambda_{\sigma,1},\lambda_{\sigma,2}:\cotangent X\rightarrow\RR$
of $\sigma$ are given by
\begin{eqnarray}
\lambda_{\sigma,1}(\xi) & = & \frac{1}{2}\tr(\sigma(\xi))-||\sigma_{0}(\xi)||,\label{eq:LamdasFormula}\\
\lambda_{\sigma,2}(\xi) & = & \frac{1}{2}\tr(\sigma(\xi))+||\sigma_{0}(\xi)||.\label{eq:LambdasFormula2}
\end{eqnarray}
\begin{figure}[h]
\subfloat[]{\includegraphics[scale=0.45]{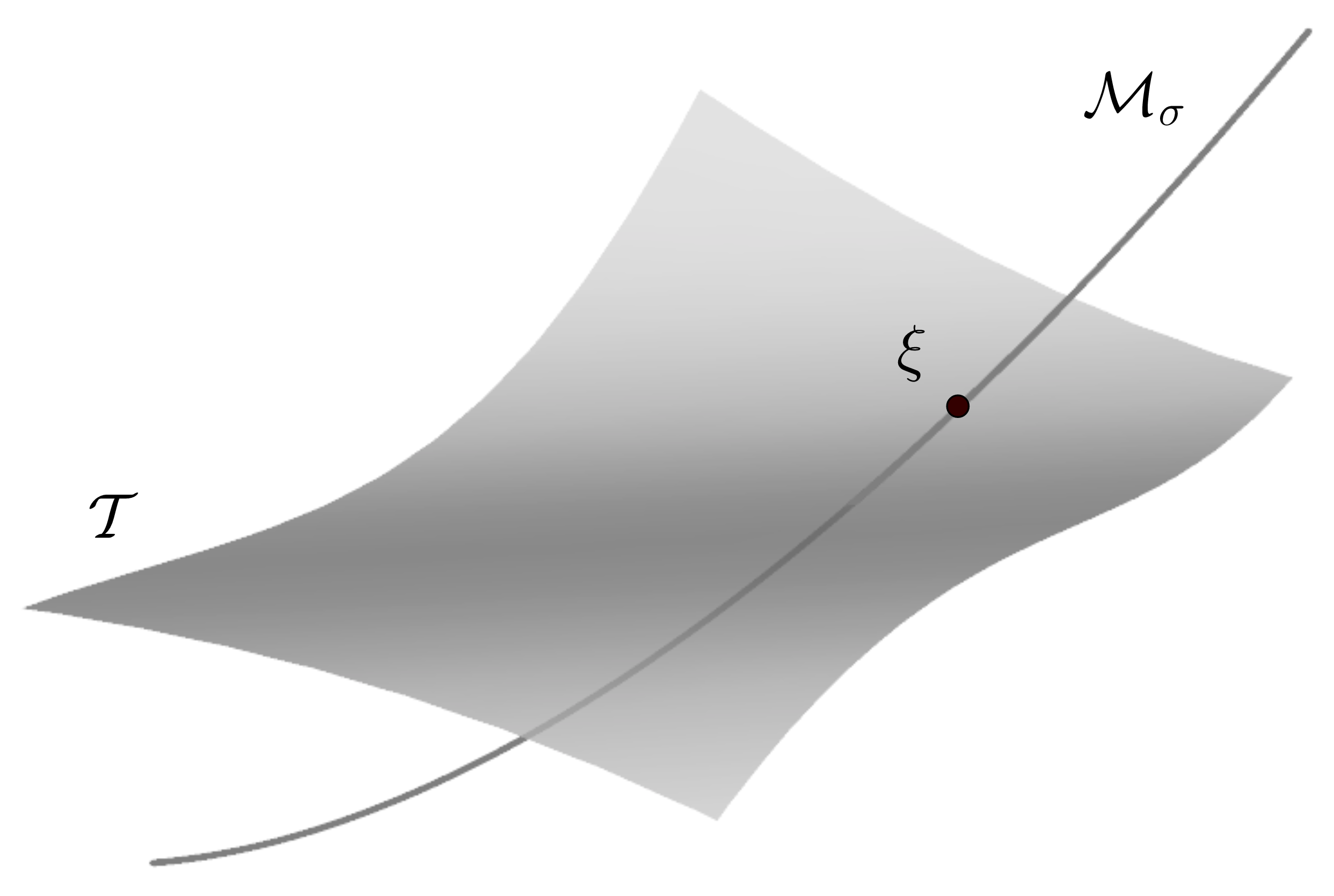}}\subfloat[]{\includegraphics[scale=0.4]{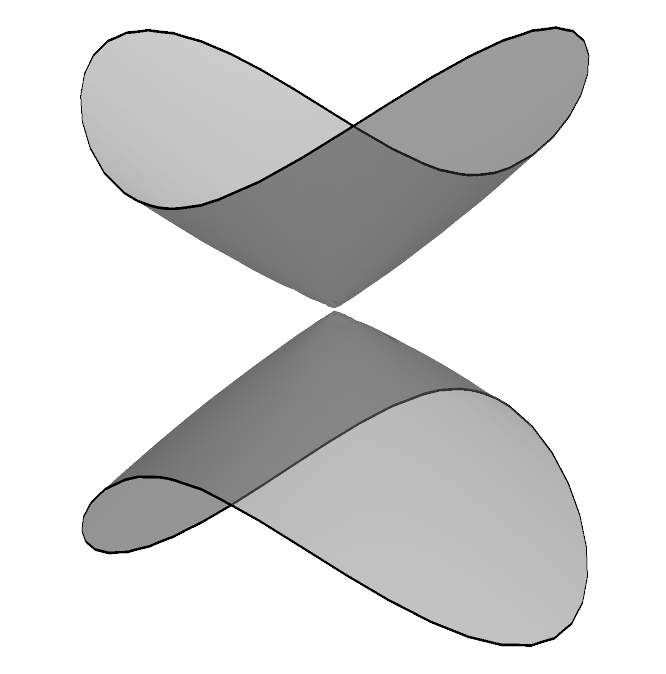}}\caption{\label{fig:Conical-singularities-of}Conical singularities of the
eigenvalue functions $\lambda_{\sigma,1}$ and $\lambda_{\sigma,2}$
of a symbol $\sigma$. On the left we show the multiplicity set $\protect\mult_{\sigma}$
and a transversal sub-manifold $\mathcal{T}$ to it at $\xi\in\protect\mult_{\sigma}$.
On the right we show the graphs near $\xi$ of $\lambda_{\sigma,1}|\mathcal{T}$
(lower cone) and $\lambda_{\sigma,2}|\mathcal{T}$ (upper cone).}
\end{figure}

\begin{defn}
The multiplicity set of a symbol $\sigma$ is the set
\[
\mult_{\sigma}=\{\xi\in S(\cotangent X)|\lambda_{\sigma,1}(\xi)=\lambda_{\sigma,2}(\xi)\}.
\]
\end{defn}

From formulas \ref{eq:LamdasFormula} and \ref{eq:LambdasFormula2}
we have that
\[
\mult_{\sigma}=\{\xi\in S(\cotangent X)|\sigma_{0}(\xi)=0\}.
\]
Let $\symz F$ stand for the bundle of traceless elements in $\sym F$.
If $\sigma_{0}|S(\cotangent X)$ is transversal to the 0-section of
$\symz F$ then $\mult_{\sigma}$ is 1-dimensional sub-manifold of
$S(\cotangent X)$. The following result shows that, for a symbol
$\sigma$ in general position, the set $\mult_{\sigma}$ can be described
as the set of elements in $S(\cotangent X)$ at which $\lambda_{\sigma,1}|S(\cotangent X)$
and $\lambda_{\sigma,2}|S(\cotangent X)$ are non-smooth (see Figure
\ref{fig:Conical-singularities-of}).
\begin{prop}
\label{pro:singularitiesequationforfresnel}Consider a symbol $\sigma$
and let $\sigma_{0}$ be its traceless part. If $\sigma_{0}|S(\cotangent X)$
is transversal to the zero section of $\symz F$ then 
\[
\mult_{\sigma}=\hbox{sing}(\lambda_{\sigma,1}|S(\cotangent X))=\hbox{sing}(\lambda_{\sigma,2}|S(\cotangent X)),
\]
where for a function $f:S(\cotangent X)\rightarrow\RR$ we have defined
\[
\hbox{sing}(f)=\{\xi\in S(\cotangent X)|f\hbox{\,\,is\,non-smooth\,at\,\,}\xi\}.
\]
\end{prop}

\begin{proof}
By formulas \ref{eq:LamdasFormula} and \ref{eq:LambdasFormula2}
it is enough to show that the map $f:S(\cotangent X)\rightarrow\RR$
given by
\[
f(\xi)=||\sigma_{0}(\xi)||
\]
is non-smooth at $\mult_{\sigma}$ and only on this set. The function
$f$ is smooth in $S(\cotangent X)-\mult_{\sigma}$ since $\sigma_{0}$
is smooth and the norm function is smooth away from the $0$-section
of $\symz F$. We will now prove that $f$ is non-smooth at $\mult_{\sigma}$.
The proof is local: consider an open set $U\subset X$ where we can
find ortho-normal sections $v,w:U\rightarrow F$. We can then write
\[
f(\xi)=|(a(\xi),b(\xi))|\text{{\,\,for\,\,}\ensuremath{\xi\in S(\cotangent X|U)}},
\]
where
\begin{align*}
a(\xi) & =<\sigma_{0}(\xi)v(\xi),v(\xi)>_{F},\\
b(\xi) & =<\sigma_{0}(\xi)v(\xi),w(\xi)>_{F},
\end{align*}
and 
\[
|(a,b)|=\sqrt{a^{2}+b^{2}}.
\]
If we define $\mult_{\sigma}|U=\mult_{\sigma}\cap S(\cotangent X|U$)
then 
\[
\mult_{\sigma}|U=\{\xi\in S(\cotangent X|U)|(a(\xi),b(\xi))=(0,0)\}.
\]
From the transversality assumptions we have that the linear functionals
$da(\xi)$ and $db(\xi)$ are linearly independent, and 
\[
T_{\xi}\mult_{\sigma}=\ker(da(\xi))\cap\ker(db(\xi))\text{{\,\,for\,\,\,}}\xi\in\mult_{\sigma}|U.
\]
For $\xi\in\mult_{\sigma}$ consider any vector $\eta\in T_{\xi}S(\cotangent X|U)$
not in $\tangent_{\xi}\mult_{\sigma}$ and a curve $\alpha=\alpha(t)$
with $\alpha(0)=\xi$ and $\dot{\alpha}(0)=\eta$. We have that
\[
\frac{f(\alpha(t))-f(\alpha(0))}{t}=\begin{cases}
\left|\frac{(a(\alpha(t)),b(\alpha(t))}{t}\right| & t>0,\\
-\left|\frac{(a(\alpha(t)),b(\alpha(t))}{t}\right| & t<0.
\end{cases}
\]
Hence
\[
\lim_{t\mapsto0^{+}}\frac{f(\alpha(t))-f(\alpha(0))}{t}=|(da(\xi)\eta,db(\xi)\eta)|
\]
and
\[
\lim_{t\mapsto0^{-}}\frac{f(\alpha(t))-f(\alpha(0))}{t}=-|(da(\xi)\eta,db(\xi)\eta)|.
\]
Since $\eta\not\in T_{\xi}\mult_{\sigma}$ we have that $|(da(\xi)\eta,db(\xi)\eta)|\not=0$
and we conclude that $f$ is not differentiable at $\xi$.
\end{proof}

\subsection{Obstruction theory for the multiplicity set}

We will now find topological obstructions for the set $\mult_{\sigma}$
to be empty.
\begin{thm}
\label{pro:SingCharFunctions}Consider a symbol $\sigma:\cotangent X\rightarrow\sym F$
and let $e(F),e(\tangent X)\in H^{2}(X,\RR)$ denote the Euler classes
of $F$ and $\tangent X$, respectively. If $e(\tangent X)-2e(F)\not=0$
in $H^{2}(X,\RR)$ then $\mult_{\sigma}\not=\emptyset$.
\end{thm}

\begin{proof}
Consider a rotation matrix 
\[
R_{\theta}=\left(\begin{array}{cc}
\cos(\theta) & -\sin(\theta)\\
\sin(\theta) & \cos(\theta)
\end{array}\right)\in\SO(2)
\]
and
\[
A=\left(\begin{array}{cc}
a & b\\
b & -a
\end{array}\right)\in\symz{\RR^{2}}.
\]
It easy to see that 
\[
R_{\theta}AR_{\theta}^{T}=\left(\begin{array}{cc}
p & q\\
q & -p
\end{array}\right),
\]
where
\[
\left(\begin{array}{c}
p\\
q
\end{array}\right)=\left(\begin{array}{cc}
\cos(2\theta) & -\sin(2\theta)\\
\sin(2\theta) & \cos(2\theta)
\end{array}\right)\left(\begin{array}{c}
a\\
b
\end{array}\right).
\]
Since the transition functions of the bundle $\symz F$ are the same
as those of the bundle $F\otimes_{\CC}F$, these bundles must be isomorphic.
If $\mult_{\sigma}$ were empty then we would have that $\sigma_{0}(\xi)\not=0$
for all $\xi\in S(\cotangent X)$, which means that $\sigma_{0}$
is an isomorphism between $\cotangent X$ and $\symz F$. We conclude
that $\tangent X$ and $F\otimes_{\CC}F$ must be isomorphic and hence
\[
e(\tangent X)=e(F\otimes_{\CC}F)=2e(F),
\]
which contradicts the hypothesis of the theorem. Hence $\mult_{\sigma}\not=\emptyset$.
\end{proof}
\begin{cor}
\label{cor:SingCharFunctions}If the genus of $X$ is different from
one then the multiplicity set of any symbol $\sigma:\cotangent X\rightarrow\sym{(\tangent X)}$
must be non-empty.
\end{cor}

\begin{proof}
This follows from the previous Proposition and the fact that 
\[
\int_{X}(e(\tangent X)-2e(TX))=2(g-1)\not=0,
\]
where $g$ is the genus of $X$.
\end{proof}

\subsection{The singular set and the kernel bundle}

In this section we will see that $\mult_{\sigma}$ (when non-empty)
is generically the intersection of $S(\cotangent X)$ with a line
bundle $\KK_{\sigma}\subset\cotangent X$ defined over a one dimensional
sub-manifold $\SS_{\sigma}$ of $X$. Recall that a symbol $\sigma:\cotangent X\rightarrow\sym F$
can be seen as a section $s_{\sigma}$ of $\hbox{Hom}(\cotangent X,\sym F)$. 
\begin{defn}
The \emph{singular set} $\SS_{\sigma}$ and the \emph{kernel bundle}
$\KK_{\sigma}$ of a symbol $\sigma$ are the sets defined by
\begin{eqnarray*}
\SS_{\sigma} & = & \left\{ x\in X|\dim\left(\ker(s_{\sigma_{0}}(x))\right)>0\right\} ,\\
\KK_{\sigma} & = & \bigcup_{x\in S_{\sigma}}\ker(s_{\sigma_{0}}(x)).
\end{eqnarray*}
\begin{figure}[h]
\includegraphics[scale=0.4]{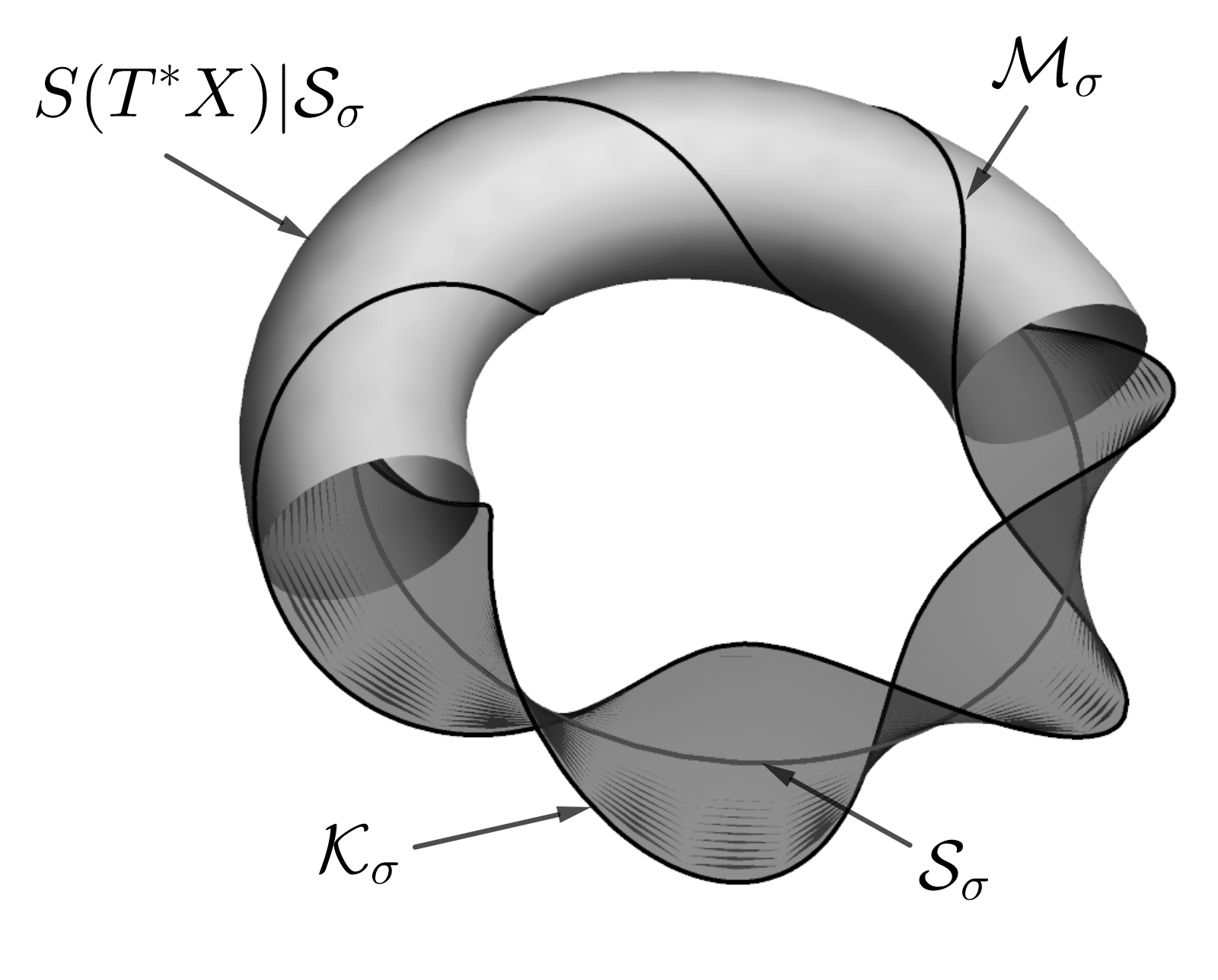}\caption{\label{fig:The-multiplicity-set}The multiplicity set $\protect\mult_{\sigma}$
as the intersection of the kernel bundle $\protect\KK_{\sigma}$ with
$S(\protect\cotangent X|\protect\SS_{\sigma})$.}
\end{figure}
The sets $\KK_{\sigma}$ and $\mult_{\sigma}$ are related as follows
(see Figure \ref{fig:The-multiplicity-set})
\[
\mult_{\sigma}=S(\cotangent X)\cap\KK_{\sigma}\hbox{\,\,and\,\,}\KK_{\sigma}=\RR\cdot\mult_{\sigma},
\]
where the action of $\lambda\in\RR$ on an element $\xi\in\mult_{\sigma}$
is that of scalar multiplication.
\end{defn}

\begin{rem}
The sets $\SS_{\sigma},\KK_{\sigma}$ and $\mult_{\sigma}$ only depend
on the traceless part of $\sigma$. Hence, when studying these spaces
we can assume that $\sigma$ is a traceless symbol.
\end{rem}

\begin{thm}
\label{thm:MultiplicityAsKernelBundle} For a symbol $\sigma$ in
general position with $\SS_{\sigma}\not=\emptyset$ we have that
\begin{enumerate}
\item The the singular set $\SS_{\sigma}$ is a smooth 1-dimensional sub-manifold
of $X$.
\item The kernel bundle $\KK_{\sigma}$ is a smooth rank one real vector
bundle over $\SS_{\sigma}$.
\item The symbol $\sigma_{0}:S(\cotangent X)\rightarrow\symz F$ is transversal
to the zero section of $\symz F$.
\end{enumerate}
\end{thm}

\begin{proof}
see section \ref{sec:Proof-Mult-As-Kernel}
\end{proof}
\begin{rem}
The importance of point 3 in the above Theorem is that for symbols
$\sigma$ in general position, the multiplicity set $\mult_{\sigma}$
can be characterized as the set where the eigenvalue functions $\lambda_{\sigma,1}|S(\cotangent X)$
and $\lambda_{\sigma,2}|S(\cotangent X)$ are non-smooth (see Proposition
\ref{pro:singularitiesequationforfresnel}).
\end{rem}

We will now show that we can write
\[
\KK_{\sigma}=L_{\sigma}|\SS_{\sigma},
\]
where $L_{\sigma}$ is line sub-bundle of $\cotangent X|(X-C_{\sigma})$
and $C_{\sigma}$ consists of a finite number of points. Recall that
we have a metric $g$ in $\symz F$ given by $g(A,B)=\tr(A\circ B)/2.$
Consider a traceless symbol $\sigma$ and let $\sigma^{*}:\symz{F\rightarrow\cotangent X}$
be its adjoint. We define the section $G_{\sigma}:X\rightarrow\sym{(\cotangent X})$
by $G_{\sigma}=s_{\sigma^{*}\circ\sigma},$ so that
\[
<G_{\sigma}(x)\xi,\eta>=g(\sigma(\xi),\sigma(\eta)).
\]
The eigenvalues of $G_{\sigma}$ define functions $\kappa_{\sigma,1},\kappa_{\sigma,2}:X\rightarrow\RR$
with corresponding orthogonal eigen-line subfields $L_{\sigma,1}$
and $L_{\sigma,2}$ of $\cotangent X|(X-C_{\sigma})$, where
\[
C_{\sigma}=\{x\in X|\kappa_{\sigma,1}(x)=\kappa_{\sigma,1}(x)\}.
\]
For symbols in general position $C_{\sigma}$ is a finite set. The
functions $\kappa_{\sigma,1}$ and $\kappa_{\sigma_{2}}$ are non-negative
and we can assume that $\kappa_{\sigma,1}\leq\kappa_{\sigma,2}$.
We can then write
\begin{equation}
0\leq\kappa_{\sigma,1}\leq\kappa_{\sigma,2}.\label{eq:KappaInequalities}
\end{equation}
By constructions we have
\[
||\sigma(\xi)||^{2}=<G_{\sigma}(x)\xi,\xi>
\]
so that for $x\in X-C_{\sigma}$ and $\xi\in S(\cotangent_{x}X)$
we have that
\begin{equation}
||\sigma(\xi)||^{2}=\kappa_{\sigma,1}(x)\cos^{2}(\theta(x))+\kappa_{\sigma,2}(x)\sin^{2}(\theta(x)),\label{eq:NormSigma0FromEigenvalues}
\end{equation}
where $\theta(x)$ is the angle that $\xi$ makes with $L_{\sigma,1}(x)$.
The above formula still make sense for $x\in C_{\sigma}$ since in
this case, independently of how we define $\theta(x)$, we have that
$||\sigma(\xi)||^{2}=\kappa_{\sigma,1}(x)=\kappa_{\sigma,2}(x)$. 
\begin{prop}
\label{prop:SsigmaFromG} For a traceless symbol $\sigma$ in general
position we have that $\kappa_{\sigma,2}>0$ and 
\begin{align*}
\SS_{\sigma} & =\{x\in X|\kappa_{\sigma,1}(x)=0\},\\
\KK_{\sigma} & =L_{\sigma,1}|\SS_{\sigma}.
\end{align*}
\end{prop}

\begin{proof}
If $x\in X$ is such that $\kappa_{2}(x)=0$ from \ref{eq:KappaInequalities}
and \ref{eq:NormSigma0FromEigenvalues} we get $\KK_{\sigma}(x)=\cotangent_{x}X$,
which contradicts Theorem \ref{thm:MultiplicityAsKernelBundle}. We
conclude that $k_{\sigma,2}(x)>0$ for all $x\in X$. Hence
\begin{align*}
\SS_{\sigma} & =\{x\in X|\kappa_{\sigma,1}(x)=0\},\\
\KK_{\sigma} & =\RR\cdot\{\xi\in S(\cotangent X)|x\in\SS_{\sigma}\hbox{\,\,and\,\,}\theta(x)\in\{0,\pi\}\}=L_{\sigma,1}(x).
\end{align*}
\end{proof}

\subsection{The degree of the multiplicity set}

\begin{figure}[h]
\subfloat[$\deg(\protect\mult_{\sigma}|C)=0$]{\includegraphics[scale=0.28]{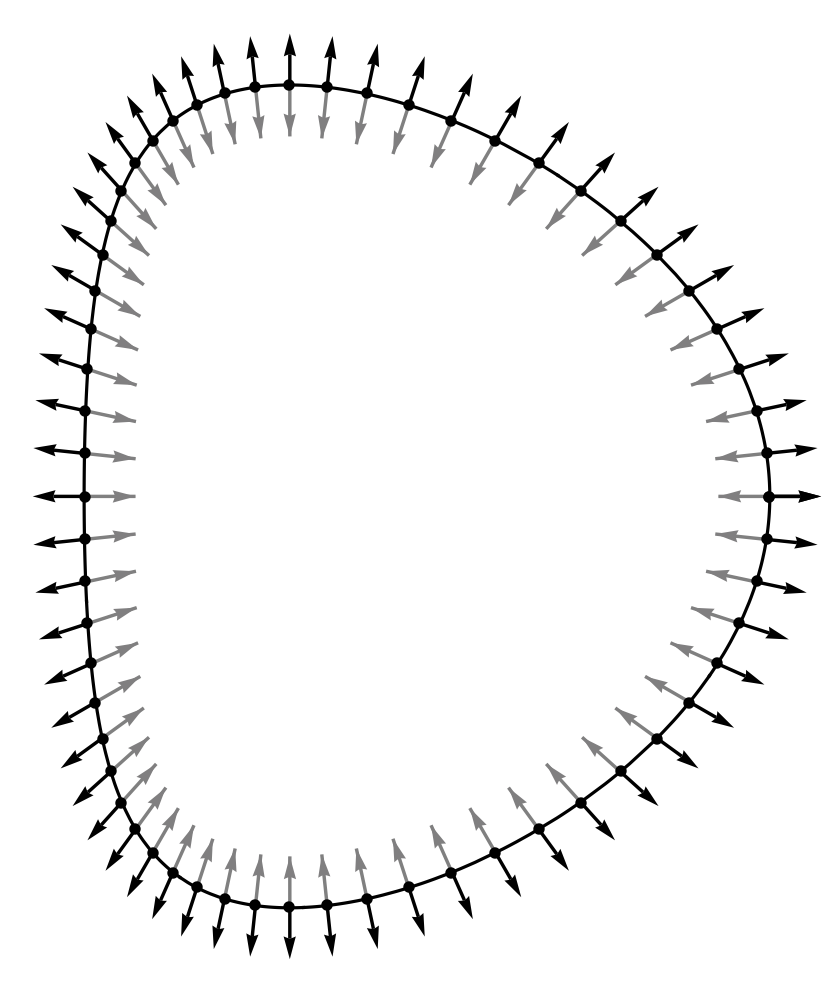}}\subfloat[$\deg(\protect\mult_{\sigma}|C)=1$]{\includegraphics[scale=0.28]{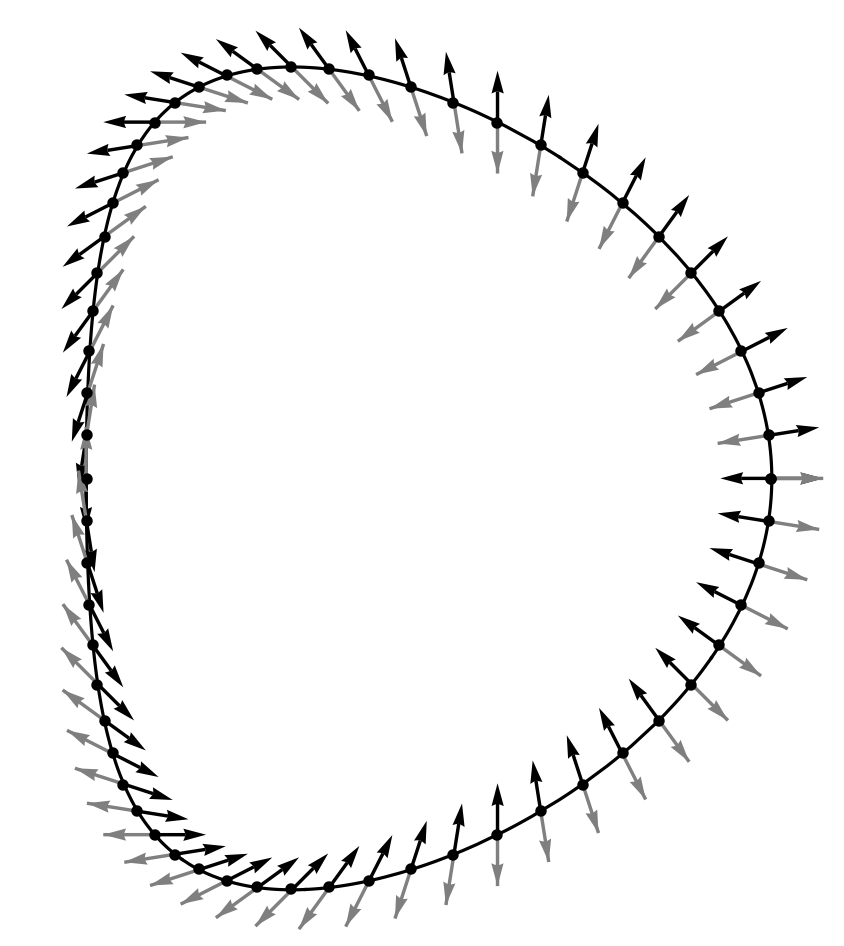}}\subfloat[$\deg(\protect\mult_{\sigma}|C)=-5$]{\includegraphics[scale=0.28]{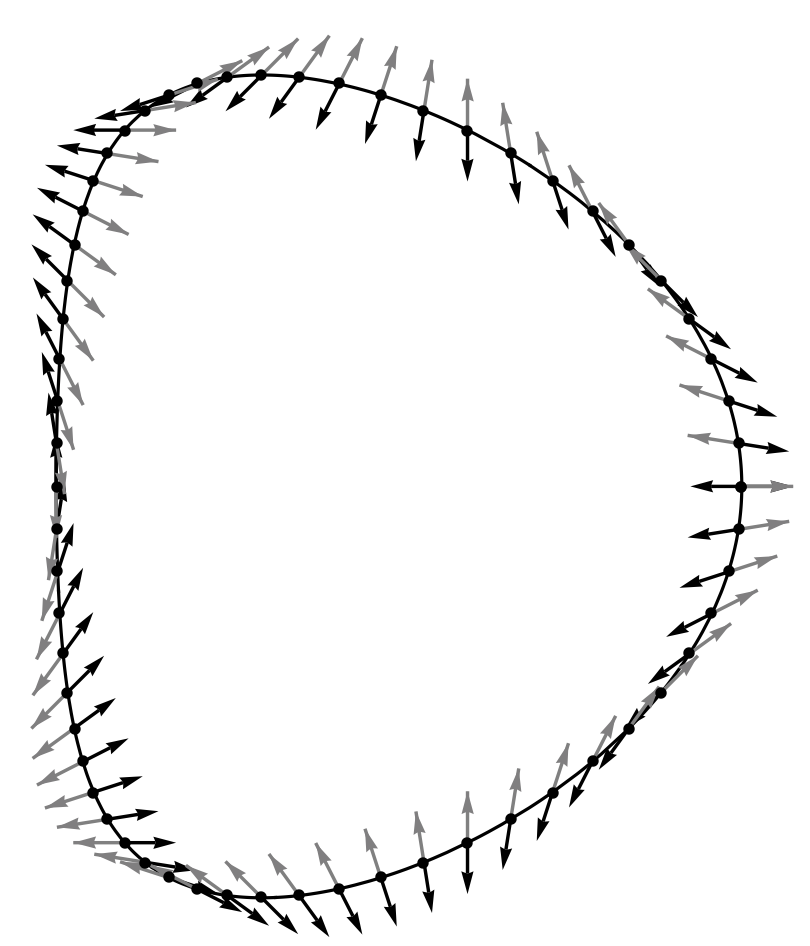}}\caption{\label{fig:The-degree-K}The degree of $\protect\mult_{\sigma}$ over
a connected component $C$ of $\protect\SS_{\sigma}$. The curve $C$
has been oriented counter-clockwise.}
\end{figure}

Consider a symbol $\sigma$ in general position with non-empty singular
set. The connected components of $\SS_{\sigma}$ are $1$-dimensional
sub-manifolds of $X$; each of them homeomorphic to a circle. An orientation
of $\SS_{\sigma}$ induces and orientation in $\mult_{\sigma}$ that
makes the projection map $\xi_{x}\mapsto x$ orientation preserving.
The degree of $\mult_{\sigma}$ over a connected component $C$ of
$\SS_{\sigma}$ is the number of times that $\KK_{\sigma}$ ``turns
with respect to $C$'' as we move along $C$ (see Figure \ref{fig:The-degree-K});
and the degree of $\mult_{\sigma}$ is the sum of these numbers over
all the connected components of $\SS_{\sigma}$. The technical details
of this definition are as follows.

Consider the map $\mu_{\sigma}:\mult_{\sigma}\rightarrow S^{1}$,
where $S^{1}=\{z\in\CC||z|=1\}$, defined by 
\[
\mu_{\sigma}(\xi_{x})=\frac{\xi_{x}^{\sharp}}{u_{\sigma}(x)},
\]
where $u_{\sigma}:\SS_{\sigma}\rightarrow\tangent\SS_{\sigma}$ is
a unit tangent field compatible with the orientation of $\SS_{\sigma}$;
the quotient $\xi_{x}^{\#}/u_{\sigma}(x)$ being defined as the unique
complex number $\exp(i\psi)\in S^{1}$ such that $\xi_{x}^{\#}=\exp(i\psi)u_{\sigma}(x)$.
\begin{defn}
For a give orientation of $\SS_{\sigma}$, the degree of $\mult_{\sigma}$
is the integer 
\[
\deg(\mult_{\sigma})=\frac{1}{2\pi i}\int_{\mult_{\sigma}}\mu_{\sigma}^{*}\left(\frac{dz}{z}\right)
\]
If we define the degree of $\mult_{\sigma}$ over a connected component
$C$ of $\SS_{\sigma}$ as
\[
\deg(\mult_{\sigma}|C)=\frac{1}{2\pi i}\int_{\mult_{\sigma}|C}\mu_{\sigma}^{*}\left(\frac{dz}{z}\right)
\]
then 
\[
\deg(\mult_{\sigma})=\text{\ensuremath{\sum_{i=1}^{k}\deg(\mult_{\sigma}|C_{\sigma,i})},}
\]
where $\{C_{\sigma,i}\}_{i=1}^{k}$ are the connected components of
$\SS_{\sigma}$.
\end{defn}

\begin{rem*}
If $d\theta$ is the angular form
\[
d\theta=\frac{-ydx}{x^{2}+y^{2}}+\frac{xdy}{x^{2}+y^{2}}
\]
we can also write
\[
\deg(\mult_{\sigma})=\frac{1}{2\pi}\int_{\mult_{\sigma}}\mu_{\sigma}^{*}(d\theta).
\]
\end{rem*}
It is important to notice that the definitions of $\deg(\mult_{\sigma})$
and $\deg(\mult_{\sigma}|C)$ depend on the choice of orientation
of $\SS_{\text{\ensuremath{\sigma}}}$; we will later see how to assign
an ``natural orientation'' to $\SS_{\sigma}$ in some cases of interest.
Also, observe that set $\mult_{\sigma}|C$ is connected if $\deg(\mult_{\sigma}|C)$
is odd and it consists of two connected components if $\deg(\mult_{\sigma}|C)$
is even.
\begin{problem}
\label{Integer realisation} Given $m\in\ZZ$ does there exists a
symbol $\sigma$ such that $\deg(\mult_{\sigma})=m$?
\end{problem}

\subsection{Physical interpretation of the degree of $\protect\mult_{\sigma}$}

\begin{figure}[h]

\includegraphics[scale=0.8]{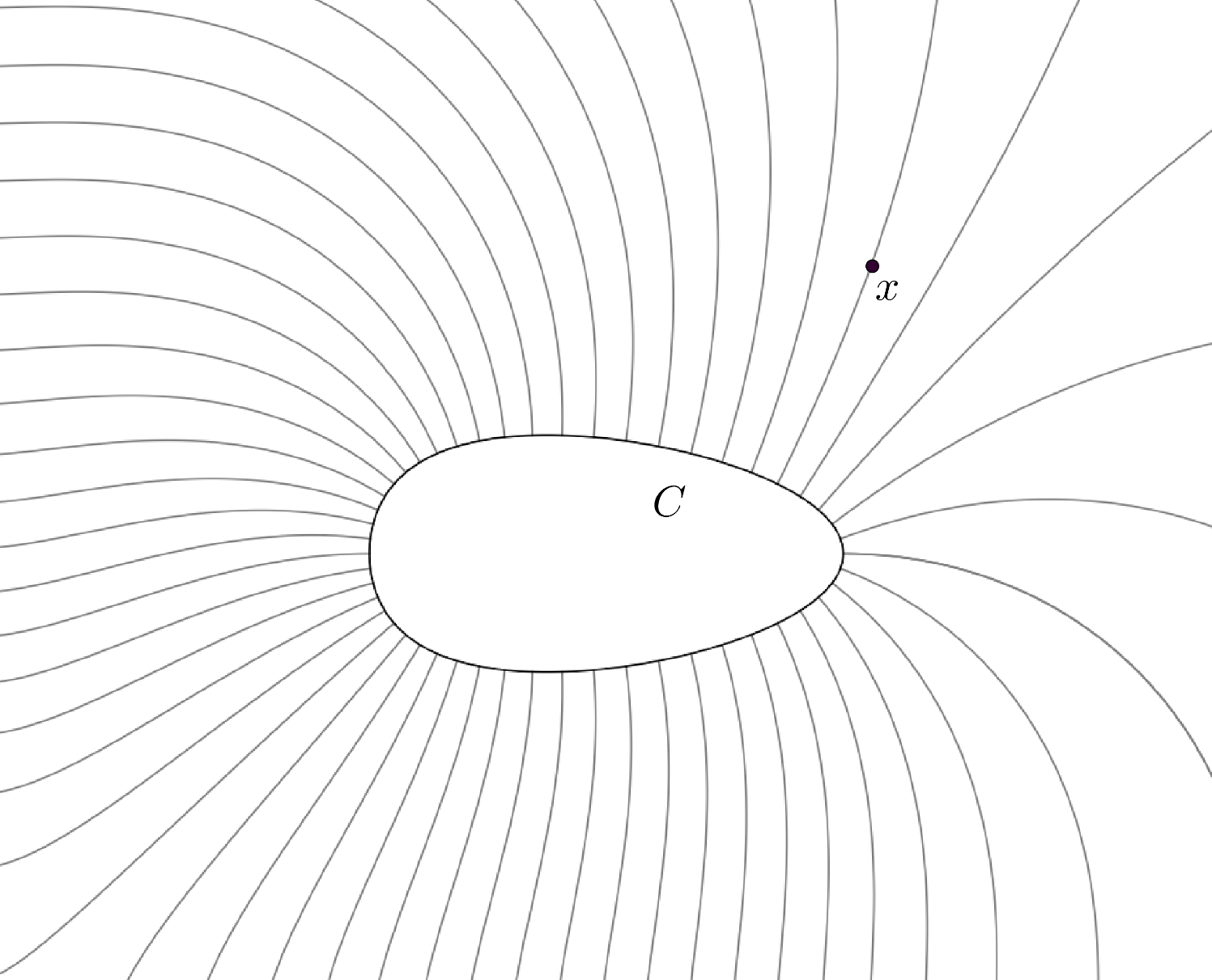}\caption{\label{fig:Rays-orthogonal-to}Rays orthogonal to $C$}

\end{figure}
Recall that when trying to find high frequency solutions to the partial
differential equation associated with $\sigma$, we arrived at the
Hamilton-Jacobi equations
\[
\der{\varphi}t(x,t)+\lambda_{\sigma,i}(d_{x}\varphi(x,t))=0,\hbox{\,\,for\,\,}i=1,2.
\]
If the phase function has the form $\varphi(x,t)=\psi(x)-t$, we obtain
\begin{equation}
\lambda_{\sigma,i}(d\psi(x))=1\hbox{\,\,for\,\,}i=1,2.\label{eq:ReducedHJ}
\end{equation}
Solutions to these equations can be obtained by the method of characteristics.
A characteristic curve is an integral line of the Hamiltonian vector
field of $\lambda_{\sigma,i}$, and the corresponding ray is its projection
onto $X$. For a $1$-dimensional submanifold $C$ of $X$ consider
the characteristic curves contained in $\{\xi\in\cotangent X|\lambda_{\sigma,i}(\xi)=1\}$
and whose corresponding rays are orthogonal to $C$ (see Figure \ref{fig:Rays-orthogonal-to}).
Solutions to \ref{eq:ReducedHJ} having $C$ as level set can be constructed
(at least near $C$)  by letting
\[
\psi(x)=\hbox{time\,of\,arrival\,of\,ray\,from\,\,}C\hbox{\,\,to\,\,}x.
\]
Consider the case when $C$ is a connected component of $\SS_{\sigma}$.
Special optical phenomena occurs when the rays emerging from $C$
are tangent to lines in $\KK_{\sigma}$. The signed count of the points
at which this happens is $\deg(\mult_{\sigma}|C)$ (see Figure \ref{fig:SingCount})
. In particular, the number of such points is bounded below by $|\deg(\mult_{\sigma}|C)|$.

\begin{figure}[h]

\includegraphics[scale=0.75]{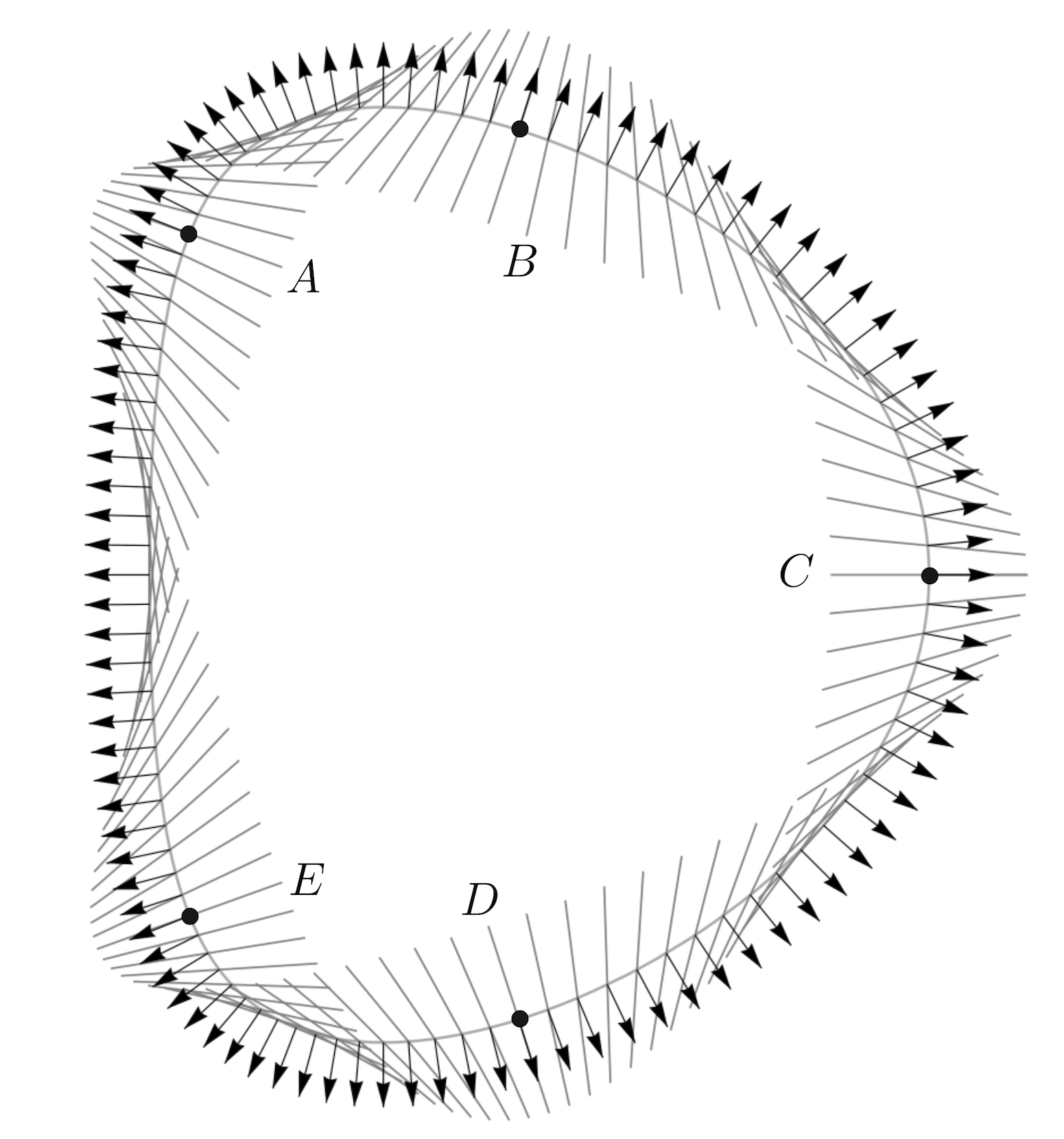}\caption{\label{fig:SingCount}The kernel bundle $\protect\KK_{\sigma}$ over
a connected component of the singular set $\protect\SS_{\sigma}$.
The lines $A,B,C,D$ and $E$ represent directions at which singular
optical phenomena occur.}

\end{figure}

\section{Symbols on the tangent bundle}

In this section we will study symbols for the case when $F=\tangent X$.
Under this assumption we will be able to give a formula to compute
the degree of the multiplicity set $\mult_{\sigma}$ in terms of the
Euler characteristics of $X$ and a sub-manifold with boundary $\NN_{\sigma}\subset X$
with $\partial\NN_{\sigma}=\SS_{\sigma}$ (see Theorem \ref{thm:MainIndexTheorem}).

\subsection{Symbols as sections of complex line bundles}

Since $F$ is an $SO(2)=U(1)$ vector bundle, we can also see it as
complex line bundle. We will let $F^{0}$ stand for the trivial complex
line bundle $X\times\CC$ and for $n>0$ define 

\begin{align*}
F^{n} & =F\otimes_{\CC}\cdots\otimes_{\CC}F,\\
F^{-n} & =\bar{F}\otimes_{\CC}\cdots\otimes_{\CC}\bar{F},
\end{align*}
where $\bar{F}$ is the conjugate bundle of $F$. If $v\in F_{x}^{n}$
and $w\in F_{x}^{m}$ then $v\otimes_{\CC}w\in F_{x}^{m+n}$, and
for simplicity we will write 
\[
vw=v\otimes_{\CC}w.
\]
Since $F$ is a one dimensional complex line bundle we have that $vw=wv$.
Furthermore, the induced norms in $F^{n}$ and $F^{m}$ are such that
$|vw|=|v||w|$. Every non-zero element $v\in F_{x}^{n}$ has an inverse
element in $v^{-1}\in F_{x}^{-n}$, which is the unique vector satisfying
$vv^{-1}=1\in F_{x}^{0}$. Hence, for $v\in F_{x}^{m}$ and $w\in F_{x}^{n}$
we can define $v/w=vw^{-1}\in F_{x}^{n-m}$. 

The following result will allow us to identify traceless symbols $\sigma$
in $\tangent X$ with two of sections $v:X\rightarrow\tangent X$
and $w:X\rightarrow(\tangent X)^{3}$. Recall that $\sym F$ has a
metric defined by
\[
g(A,B)=\frac{1}{2}\tr(A\circ B)\hbox{\,\,for\,\,}A,B\in\sym{F_{x}}.
\]

\begin{thm}
\label{thm:ISOMORPHISM}Let $F$ be an $\SO(2)$ vector bundle of
rank $2$ over $X$. There exists an isomorphism of real vector bundles
$\Phi:F\oplus F^{3}\rightarrow\hbox{Hom}(F^{*},\symz F)$ such that
if for $v\in F_{x}$ and $w\in F_{x}^{3}$ we let $\sigma=\Phi(v\oplus w)$,
then the endomorphism $G_{\sigma}:F_{x}^{*}\rightarrow F_{x}^{*}$
defined by
\[
<G_{\sigma}\xi,\eta>=g(\sigma(\xi),\sigma(\eta)).
\]
 has eigenvalues
\begin{align*}
\kappa_{1} & =\frac{1}{2}\left(|v|-|w|\right)^{2},\\
\text{\ensuremath{\kappa_{2}}} & =\frac{1}{2}\left(|v|+|w|\right)^{2}.
\end{align*}
Furthermore, for $v\not=0$ and $w\not=0$ the eigen-spaces $L_{1}$
and $L_{2}$ corresponding to $\kappa_{1}$ and $\kappa_{2}$ are
\begin{align*}
L_{1} & =\RR\cdot\{z^{\flat}\in F_{x}^{*}|z^{2}=w/v\},\\
L_{2} & =\RR\cdot\{z^{\flat}\in F_{x}^{*}|z^{2}=-w/v\}.
\end{align*}
 
\end{thm}

\begin{proof}
See section \ref{sec:Proof-ISOMORPHISM}
\end{proof}
\begin{cor}
\label{cor:SymbolsFromComplex} For any traceless symbol $\sigma:\cotangent X\rightarrow\symz{(\tangent X)}$
we can find unique sections $v:X\rightarrow\tangent X$ and $w:X\rightarrow(\tangent X)^{3}$
such that
\[
\sigma=\sigma_{s}\hbox{\,\,\,where\,\,\,}s(x)=\Phi(v(x)\oplus w(x)).
\]
In this case we have that
\[
\SS_{\sigma}=\{x\in X||w(x)/v(x)|=1\},
\]
and 
\[
\mult_{\sigma}=\bigcup_{x\in\SS_{\sigma}}\left\{ z^{\flat}\in\cotangent_{x}X|z\in\tangent_{x}X\hbox{\,\,and\,\,}z^{2}=\frac{w(x)}{v(x)}\right\} .
\]
\end{cor}

\begin{proof}
The result follows directly from Theorem \ref{thm:ISOMORPHISM} and
Proposition \ref{prop:SsigmaFromG} by letting $v$ and $w$ be given
by
\[
v(x)\oplus w(x)=\Phi^{-1}(s_{\sigma}(x)),
\]
\end{proof}
\begin{defn}
For sections $v:X\rightarrow\tangent X$ and $w:X\rightarrow(\tangent X)^{3}$
we will refer to the traceless symbol $\sigma_{v,w}:\cotangent X\rightarrow\symz{(\tangent X)}$
given by
\[
\sigma_{v,w}=\sigma_{s}\hbox{\,\,\,for\,\,\,}s(x)=\Phi(v(x)\oplus w(x)),
\]
as the symbol associated $v$ and $w$. 
\end{defn}

\begin{rem}
For the symbol $\sigma=\sigma_{v,w}$ formula \ref{eq:NormSigma0FromEigenvalues}
becomes
\[
||\sigma(\xi_{x})||^{2}=\frac{1}{2}\left(|v(x)|^{2}+|w(x)|^{2}-2|v(x)||w(x)|\cos(2\theta(x))\right).
\]
\end{rem}

\subsection{A formula for the degree of $\protect\mult_{\sigma}$ }

If $C$ is a 1-dimensional compact oriented sub-manifold of $X$ and
$v_{1},v_{2}:C\rightarrow F^{n}|C$ are non-vanishing sections, we
define the degree of $v_{2}$ with respect to $v_{1}$ as the integer
\[
\deg_{v_{1}}(v_{2})=\frac{1}{2\pi i}\int_{C}\mu^{*}\left(\frac{dz}{z}\right),
\]
where $\mu:C\rightarrow\CC-\{0\}$ is given by
\[
\mu(x)=\frac{v_{2}(x)}{v_{1}(x)}.
\]
If $\{C_{k}\}_{1\leq i\leq n}$ are the connected components of $C$
and $\{\gamma_{k}:[a_{k},b_{k}]\rightarrow C_{k}\}_{1\leq k\leq n}$
are orientation preserving parametrization maps, then 
\[
\deg_{v_{1}}(v_{2})=\frac{1}{2\pi i}\sum_{k=1}^{n}\int_{a_{k}}^{b_{k}}\left(\frac{\dot{\mu}_{k}(t)}{\mu_{k}(t)}\right)dt\hbox{\,\,where\,\,}\text{\ensuremath{\mu_{k}(t)=\mu(\gamma_{k}(t)).}}
\]
It is easy to verify that if $v_{1},v_{2}$ are non-vanishing sections
of $F^{n}$ and $w_{1},w_{2}$ are non-vanishing sections of $F^{m}$,
then
\[
\deg_{v_{1}w_{1}}(v_{2}w_{2})=\deg_{v_{1}}(v_{2})+\deg_{w_{1}}(w_{2}).
\]

\begin{defn}
\label{def:DegreeRespectToCurve}For a non-vanishing section $v:C\rightarrow(\tangent X)^{n}|C$
the degree of $v$ with respect to $C$ is the integer
\[
\deg_{C}(v)=\deg_{u^{n}}(v),
\]
 where $u:C\rightarrow TC$ is a unit tangent field compatible with
the orientation of $C$.
\end{defn}

\begin{prop}
\label{prop:IndexMultFromFields}Consider sections $v:X\rightarrow\tangent X$
and $w:X\rightarrow(\tangent X)^{3}$ with no common zeros. If $\sigma=\sigma_{v,w}$
and $C$ is a connected component of $\SS_{\sigma}$ then 
\[
\deg(\mult_{\sigma}|C)=\deg_{C}(w)-\deg_{C}(v),
\]
and hence
\[
\deg(\mult_{\sigma})=\deg_{\SS_{\sigma}}(w)-\deg_{\SS_{\sigma}}(v).
\]
\end{prop}

\begin{proof}
The set $\mult_{\sigma}$ consists of the elements of the form $z^{\flat}(x)\in\cotangent X|\SS_{\sigma}$
such that $z(x)\in\tangent X|\SS_{\sigma}$ satisfies
\[
z(x)=\pm\left(\frac{w(x)}{v(x)}\right)^{1/2}.
\]
This formula can be re-written as
\[
\frac{z(x)}{u(x)}=\pm(h(x))^{1/2}
\]
for $h:\SS_{\sigma}\rightarrow\CC-\{0\}$ given by
\[
h(x)=\left(\frac{v(x)}{u(x)}\right)^{-1}\left(\frac{w(x)}{u^{3}(x)}\right),
\]
where $u(x)\in T_{x}\SS_{\sigma}$ is unitary and compatible with
the orientation of $\SS_{\sigma}$. If $\gamma:[a,b]\rightarrow C$
is an orientation preserving arc-length parametrization of $C$ then
$u(t)=\dot{\gamma}(t)$. If we write $z(t)=z(\gamma(t)),v(t)=v(\gamma(t)),w(t)=w(\gamma(t)),h(t)=h(\gamma(t))$,
and let $\psi,\theta:[a,b]\rightarrow\RR$ be continuous functions
such that 
\[
\frac{w(t)}{u^{3}(t)}=|w(t)|\exp(i\psi(t))\hbox{\,\,and\,\,\,\,}\frac{v(t)}{u(t)}=|v(t)|\exp(i\theta(t)),
\]
then
\begin{align*}
\deg(w|C) & =\frac{\psi(b)-\psi(a)}{2\pi},\\
\deg(v|C) & =\frac{\theta(b)-\theta(a)}{2\pi}.
\end{align*}
Since $\gamma(t)\in\SS_{\sigma}$, we have that $|w(t)|=|v(t)|$ and
hence
\[
h(t)=\exp\left(i\varphi(t)\right)\hbox{\,\,where\,\,}\varphi(t)=\psi(t)-\theta(t).
\]
From the above we conclude that
\[
\frac{z(t)}{u(t)}=\exp(i\varphi(t)/2)\hbox{\,\,or\,\,\,}\frac{z(t)}{u(t)}=\exp(i(\varphi(t)+\pi)/2),
\]
and hence
\begin{align*}
\deg(\mult_{\sigma}|C) & =\frac{1}{2\pi}\left(\left(\frac{\varphi(b)-\varphi(a)}{2}\right)+\left(\frac{(\varphi(b)+\pi)-(\varphi(a)+\pi)}{2}\right)\right)\\
 & =\frac{\varphi(b)-\varphi(a)}{2\pi}\\
 & =\frac{\psi(b)-\psi(a)}{2\pi}-\frac{\theta(b)-\theta(a)}{2\pi}\\
 & =\deg(w|C)-\deg(v|C).
\end{align*}
The second part of the proposition is a direct consequence of the
first part.
\end{proof}
We will denote the set of zeros of a section $v:X\rightarrow(TX)^{n}$
as $Z_{v}$. If $x\in X$ is an isolated zero of $v$, the degree
of $v$ at $x$ can be computed as
\[
\deg_{x}(v)=\deg_{w|\partial D}(v),
\]
where $D$ is a ``small disk'' with $Z_{v}\text{\ensuremath{\cap D=\{x\}}}$
and $w:D\rightarrow(\tangent X)^{n}|D$ is a non-vanishing vector
field. 

\begin{prop}
\label{prop:IndexFromLocalDegrees-1}Let $Y\subset X$ be a 2-dimensional
manifold with boundary $\partial Y$, and consider sections $u,v:Y\rightarrow(\tangent X)^{n}|Y$
whose zero sets $Z_{u}$ and $Z_{v}$ are such that $(Z_{u}\cup Z_{v})\cap Y$
is finite and $Z_{u}\cap\partial Y=Z_{v}\cap\partial Y=\emptyset$.
We have that 
\[
\deg_{u|\partial Y}(v)=\sum_{x\in Z_{v}\cap Y}\deg_{x}(v)-\sum_{x\in Z_{u}\cap Y}\deg_{x}(u).
\]
\end{prop}

\begin{proof}
Let $Y_{0}$ be the set obtained from $Y$ by removing a family of
small open disks $\{D_{x}\}_{x\in(Z_{u}\cup Z_{v})\cap Y}$, where
each $D_{x}$ contains $x$ and each $D_{x}$ is contained in the
interior of $Y$. Consider the map $\mu:Y_{0}\rightarrow\CC-\{0\}$
defined by
\[
\mu(x)=\frac{v(x)}{u(x)}.
\]
Using Stokes Theorem and the fact that the form $dz/z$ is closed
in $\CC-\{0\}$, we obtain
\[
\int_{\partial Y_{0}}\mu^{*}\left(\frac{dz}{z}\right)=0.
\]
From this and the formula
\[
\partial Y_{0}=\partial Y-\sum_{x\in(Z_{u}\cup Z_{v})\cap Y}\partial D_{x},
\]
we obtain
\[
\deg_{u|\partial Y}(v)=\frac{1}{2\pi i}\sum_{x\in(Z_{u}\cup Z_{v})\cap Y}\int_{\partial D_{x}}\mu^{*}\left(\frac{dz}{z}\right).
\]
Since we can write
\[
\frac{v}{u}=\frac{v/w_{x}^{n}}{u/w_{x}^{n}},
\]
for $w_{x}:D_{x}\rightarrow TX|D_{x}$ a non-vanishing vector field,
we have
\[
\frac{1}{2\pi i}\int_{\partial D_{x}}\mu^{*}\left(\frac{dz}{z}\right)=\deg_{x}(v)-\deg_{x}(u).
\]
We conclude that
\[
\deg_{u|\partial Y}=\frac{1}{2\pi i}\left(\sum_{x\in(Z_{u}\cup Z_{v})\cap Y}\deg_{x}(v)-\sum_{x\in(Z_{u}\cup Z_{v})\cap Y}\deg_{x}(u)\right).
\]
The result of the Proposition then follows from the formulas
\begin{align*}
\sum_{x\in(Z_{u}\cup Z_{v})\cap Y}\deg_{x}(v) & =\sum_{x\in Z_{v}\cap Y}\deg_{x}(v)+\sum_{x\in(Z_{u}-Z_{v})\cap Y}\deg_{x}(v),\\
\sum_{x\in(Z_{u}\cup Z_{v})\cap Y}\deg_{x}(u) & =\sum_{x\in Z_{u}\cap Y}\deg_{x}(u)+\sum_{x\in(Z_{v}-Z_{u})\cap Y}\deg_{x}(u),
\end{align*}
and observing that $\deg_{x}(v)=0$ for $x\in Z_{u}-Z_{v}$ and $\deg_{x}(u)=0$
for $x\in Z_{v}-Z_{u}$.
\end{proof}
\begin{rem}
When using Stoke's Theorem in the proof of the above Proposition we
implicitly assumed that $\partial Y$ is oriented by the field $u=in$,
where $n:\partial Y\rightarrow\tangent X|\partial Y$ is the unit
normal field to $\partial Y$ that points to the outside of $Y$.
\end{rem}

\begin{cor}
\label{cor:MainCorollary}Let $f:X\rightarrow\RR$ be a smooth function
having $0$ as a regular value, and let
\[
Y=\{x\in X|f(x)\leq0\}.
\]
If $v:X\rightarrow(\tangent X)^{n}$ does not vanish in $\partial Y$
and has a finite number of zeros in $Y,$ then
\[
\deg_{\partial Y}(v)=\sum_{x\in Z_{v}\cap Y}\deg_{x}(v)-n\chi(Y).
\]
\end{cor}

\begin{proof}
Let $g$ be a small perturbation of $f$ that makes $g$ a Morse function.
The gradient field $\nabla f$ points to the outside of $Y$. By choosing
$g$ close enough to $f$ we can ensure that $\nabla g$ will also
point to the outside of $Y$ and 
\[
\deg_{(\nabla f)^{n}|\partial Y}(v)=\deg_{(\nabla g)^{n}|\partial Y}(v).
\]
Using this formula and the fact that $i\nabla f/|\nabla f|$ is the
unit tangent field to $\partial Y$ compatible with its orientation,
we obtain
\begin{align*}
\deg_{\partial Y}(v) & =\deg_{(i\nabla f/|\nabla f|)^{n}|\partial Y}(v)\\
 & =\deg_{(\nabla f)^{n}|\partial Y}(v)\\
 & =\deg_{(\nabla g)^{n}|\partial Y}(v).
\end{align*}
Since $g$ is a Morse function the set of critical points $C_{g}$
of $g$ must be finite. Furthermore, since $g$ is close to $f$ and
$0$ is a regular value of $f$, we have that none of the zeros of
$\nabla g$ can be in $\partial Y.$ Hence, we can apply Proposition
\ref{prop:IndexFromLocalDegrees-1} to obtain
\begin{align*}
\deg_{\partial Y}(v) & =\sum_{x\in Y\cap Z_{v}}\deg_{x}(v)-\sum_{x\in Y\cap C_{g}}\deg_{x}((\nabla g)^{n})\\
 & =\sum_{x\in Y\cap Z_{v}}\deg_{x}(v)-n\sum_{x\in Y\cap C_{g}}\deg_{x}(\nabla g)
\end{align*}
The result now follows from the fact that $\nabla g$ points to the
outside of $Y$, so that (see \cite{kn:milnor_difftop})
\[
\sum_{x\in Y\cap C_{g}}\deg_{x}(\nabla g)=\chi(Y).
\]
\end{proof}
\begin{thm}
\label{thm:MainIndexTheorem}Let $\sigma=\sigma_{v,w}$ be a traceless
symbol in general position. The set 
\[
\NN_{\sigma}=\{x\in X||w(x)/v(x)|\leq1\}
\]
is a smooth manifold with boundary and 
\[
\deg(\mult_{\sigma})=3\chi(X)-2\chi(\NN_{\sigma}).
\]
\end{thm}

\begin{proof}
For symbols in general position the function
\[
f(x)=\left|\frac{w(x)}{v(x)}\right|-1
\]
has $0$ as a regular value, which implies that $\NN_{\sigma}$ is
a smooth manifold with boundary. Using Proposition \ref{prop:IndexMultFromFields}
we obtain
\[
\deg(\mult_{\sigma})=\deg_{\SS_{\sigma}}(w)-\deg_{\SS_{\sigma}}(v).
\]
Since $\SS_{\sigma}=\partial\NN_{\sigma}$, by Corollary \ref{cor:MainCorollary}
we have that
\begin{align*}
\deg_{\SS_{\sigma}}(w) & =\sum_{x\in Z_{w}\cap\NN_{\sigma}}\deg_{x}(w)-3\chi(\NN_{\sigma}),\\
\deg_{\SS_{\sigma}}(v) & =\sum_{x\in Z_{v}\cap\NN_{\sigma}}\deg_{x}(v)-\chi(\NN_{\sigma}).
\end{align*}
For symbols in general position $Z_{v}\cap Z_{w}=\emptyset$ and hence
$Z_{v}\cap\NN_{\sigma}=\emptyset$, since $v(x)=0$ for $x\in\NN_{\sigma}$
would imply that $w(x)=0$. We conclude that
\[
\deg(\mult_{\sigma})=\sum_{x\in Z_{w}\cap\NN_{\sigma}}\deg_{x}(w)-2\chi(\NN_{\sigma})
\]
Using the fact that $Z_{w}\subset\NN_{\sigma}$, we obtain 
\[
\sum_{x\in Z_{w}\cap\NN_{\sigma}}\deg_{x}(w)=\sum_{x\in Z_{w}}\deg_{x}(w)=\int_{X}e\left((TX)^{3}\right)=3\chi(X).
\]
Hence
\[
\deg(\mult_{\sigma})=3\chi(X)-2\chi(\NN_{\sigma}).
\]
\end{proof}
\begin{rem}
From the proof of the above Theorem we can see that if $D$ is a connected
component of $\NN_{\sigma}$, then
\begin{equation}
\deg(\mult_{\sigma}|\partial D)=\left(\sum_{x\in Z_{v}\cap D}\deg_{x}(w)\right)-2\chi(D).\label{eq:MainFormulaPerComponent}
\end{equation}
\end{rem}

\section{Examples}

\subsection{Holomorphic symbols on the complex plane}

\begin{figure}[h]

\subfloat[\label{fig:IndexAnnulus}$r_{0}=1$]{\includegraphics[scale=0.25]{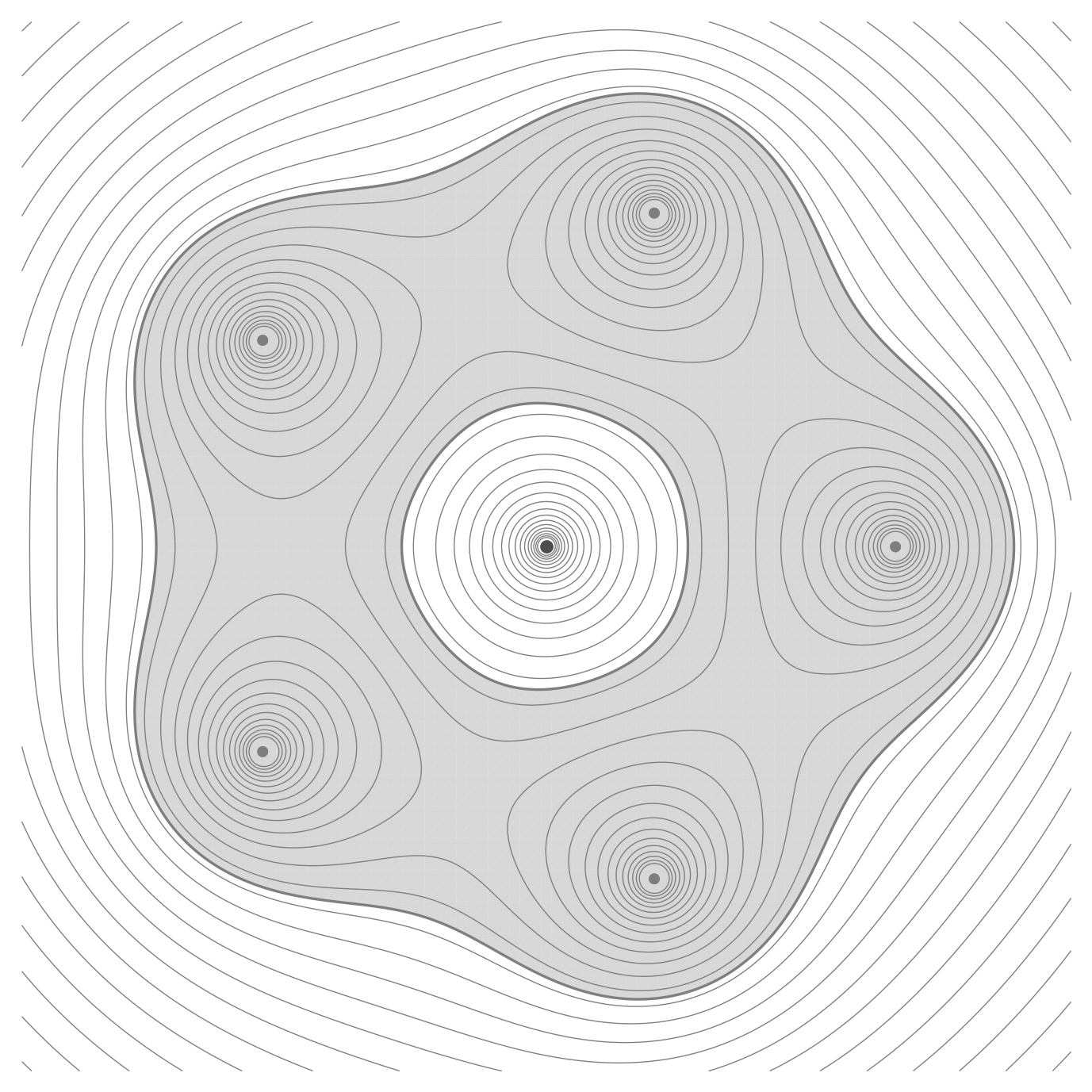}}\subfloat[\label{fig:Index5Discs}$r_{0}=1/3$]{\includegraphics[scale=0.25]{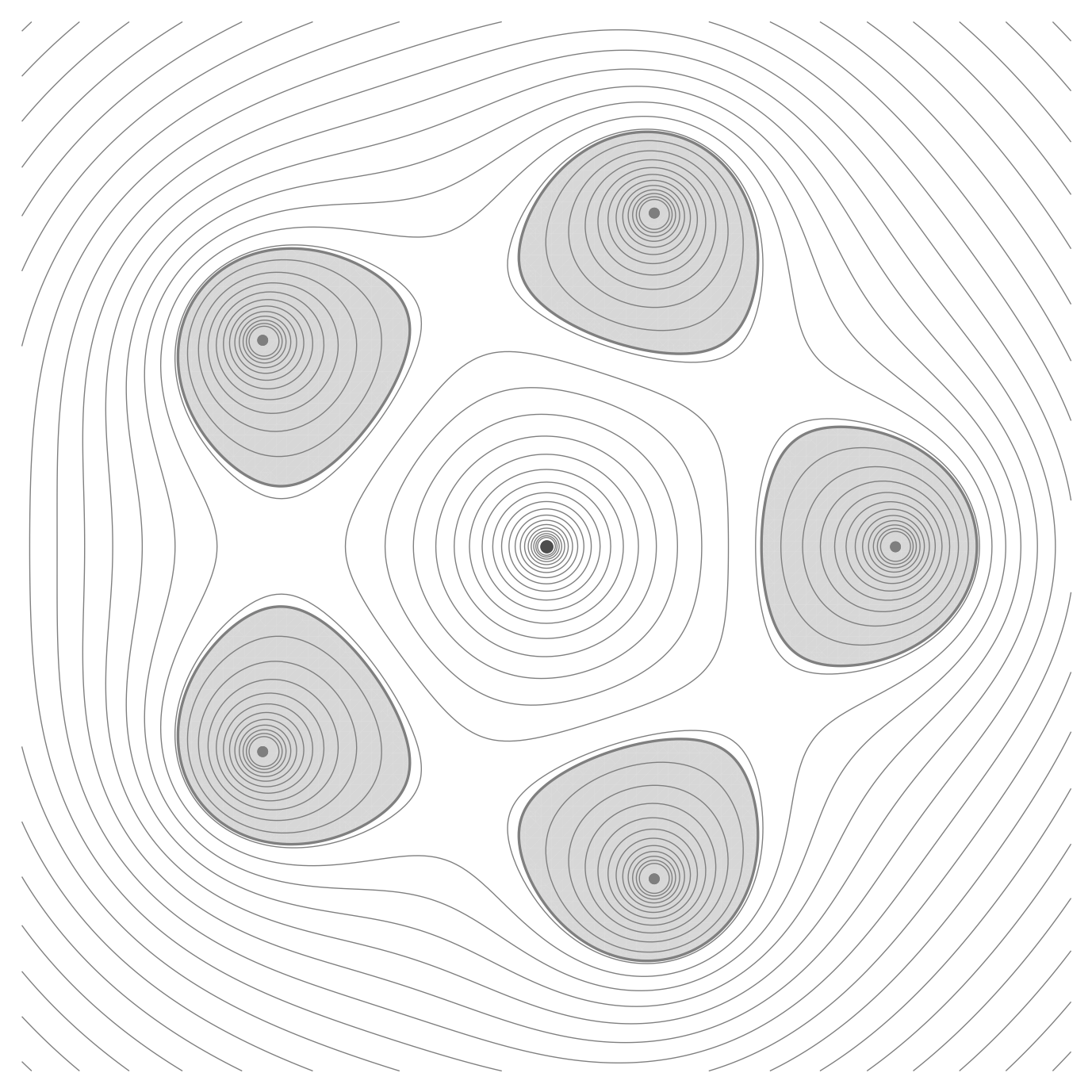}}\caption{\label{fig:PositiveChargesPlane} Equipotential curves for five negative
charges at the $5$-th roots of unity, and one positive charge at
$0$. The shaded regions represent the set $\protect\NN_{\sigma}$
for the given value of $r_{0}$.}
\end{figure}
Although we have assumed throughout the paper that the surface $X$
is compact, it will be instructive to study symbols on the plane $\RR^{2}=\CC$.
We can write the traceless part of such symbols as $\sigma=\sigma_{v,w}$
for
\[
v(z)=a(z)\dero x\and w(z)=b(z)\left(\dero x\right)^{3},
\]
where $a,b$ and complex valued functions on the complex plane. Since
we are using the flat metric in $\CC=\RR^{2}$ we have that
\[
|v(z)|=|a(z)|\text{\hbox{\,\,and\,\,}}|w(z)|=|b(z)|,
\]
and hence
\[
\SS_{\sigma}=\{z\in\CC||b(z)/a(z)|=1\}.
\]
If $a$ and $b$ are complex polynomials
\[
a(z)=a_{0}\Pi_{i=1}^{n_{a}}(z-a_{i})\and b(z)=b_{0}\Pi_{j=1}^{n_{b}}(z-b_{i})
\]
then $\SS_{\sigma}$ consists of the points $z\in\CC$ that satisfy
the equation
\begin{equation}
\sum_{i=1}^{n_{a}}\log\left(\frac{1}{|z-a_{i}|}\right)-\sum_{j=1}^{n_{b}}\log\left(\frac{1}{|z-b_{i}|}\right)=r_{0},\label{eq:Equipotential}
\end{equation}
where
\[
r_{0}=\log\left(\frac{|a_{0}|}{|b_{0}|}\right).
\]
The harmonic function $z\mapsto\log(1/|z-p|)$ represents the potential
of a positively charged particle at $p$. We conclude that $\SS_{\sigma}$
is an equipotential curve of a superposition of positively and negatively
charged particles. The negative charges are inside of $\NN_{\sigma}$
and the positive charges outside of it.
\begin{example}
Let us assume there are no positive charges for a $n\geq1$ let $b(z)=z^{n}$.
For $r_{0}=0$ the set $\NN_{\sigma}$ is the unit disk, and hence
$\chi(\NN_{\sigma})=1$. Since $\deg_{0}(b)=n$, using Formula \ref{eq:MainFormulaPerComponent}
we obtain
\[
\deg(\mult_{\sigma})=n-2.
\]
\begin{example}
Consider the case where we have $5$ negative charges located at $5$-th
roots of unity and one positive charge at $0$. For $r_{0}=1$ the
set $\NN_{\sigma}$ is a topological annulus (see Figure \ref{fig:IndexAnnulus})
so that $\chi(\NN_{\sigma})=0$, and hence
\[
\deg(\mult_{\sigma})=5-2\chi(\NN_{\sigma})=5.
\]
For $r_{0}=1/3$ the set $\NN_{\sigma}$ consists of 5 topological
discs (see Figure \ref{fig:Index5Discs}) so that $\chi(\NN_{\sigma})=5$,
and hence
\[
\deg(\mult_{\sigma})=5-2\chi(\NN_{\sigma})=-5.
\]
\end{example}

\end{example}

\subsection{Holomorphic symbols on the sphere}

\begin{figure}[h]
\includegraphics[scale=0.35]{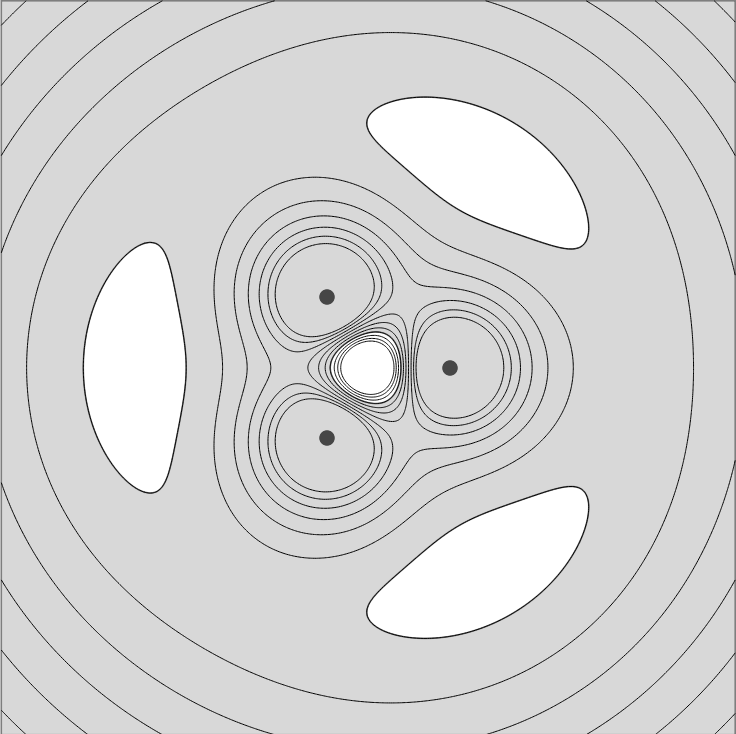}\caption{\label{fig:NsigmaSphere}Equipotential lines of 3 negatively charges
on the sphere. The shaded region represents the set $\protect\NN_{\sigma}\subset S^{2}$.}
\end{figure}
We consider the sphere as the Riemann sphere. The vector fields in
$\tangent S^{2}$ and $(\tangent S^{2})^{3}$ in $z$-coordinates
have the form 
\[
v(z)=a(z)\dero x\hbox{\,\,and\,\,}w(z)=b(z)\left(\dero x\right)^{3}.
\]
The metric in $S^{2}$ is such that
\[
\left|\dero x\right|=\frac{2}{1+|z|^{2}}\hbox{\,\,\,for\,\,\,}\dero x\in\tangent_{z}S^{2},
\]
and hence 
\[
|v(z)|=\frac{2|a(z)|}{1+|z|^{2}}\hbox{\,\,and\,\,}|w(z)|=\frac{8|b(z)|}{(1+|z|^{2})^{3}}.
\]
We conclude that $\SS_{\sigma}$ consists of the points $z$ satisfying
\begin{equation}
\left(\frac{4}{(1+|z|^{2})^{2}}\right)\frac{|b(z)|}{|a(z)|}=1.\label{eq:SingSetSphere}
\end{equation}
If we denote the complex variables at $0$ and $\infty$ by $z=x+iy$
and $w=u+iv$, the change of variables formula is 
\[
w=\frac{1}{z},
\]
so that

\[
\dero x=-z^{-2}\dero u.
\]
From this it follows that for any quadratic polynomial $P(z)=a_{0}+a_{1}z+a_{2}z^{2}$
the field
\[
z\mapsto P(z)\dero x=(a_{0}+a_{1}z+a_{2}z^{2})\dero x
\]
extends to a vector field over the whole sphere $S^{2}$. In fact,
in $w$ coordinates the above field is written as
\[
w\mapsto-w^{2}P(1/w)\dero u=-(a_{2}+a_{1}w+a_{0}w^{2})\dero u.
\]
Suppose that $a$ is a polynomial of degree $n_{a}$ and $b$ a polynomial
of degree $n_{b}$. From the above discussion we conclude that for
$0\leq n_{a}\leq2$ and $0\leq n_{b}\leq6$ the sections
\[
z\mapsto a(z)\dero x\and z\mapsto b(z)\left(\dero x\right)^{3}
\]
correspond to a traceless symbols on $\tangent S^{2}$. Equation \ref{eq:SingSetSphere}
for the singularity can be written 
\begin{equation}
2\log\left(\frac{2}{1+|z|^{2}}\right)+\sum_{i=1}^{n_{a}}\log\left(\frac{1}{|z-a_{i}|}\right)-\sum_{i=1}^{n_{b}}\log\left(\frac{1}{|z-b_{i}|}\right)=r_{0}\label{eq:SinSetSphereLog}
\end{equation}
where $r_{0}=\log(|a_{0}|/|b_{0}|)$. This formula is very similar
to the one we have studied for symbols in the plane, but for the appearance
of the first term which accounts for the fact that we are working
on the sphere.
\begin{example}
Let us assume that there are no positive charges and that $b(z)=z^{n}$.
Some simple calculations show that for for $n=1$ the set $\NN_{\sigma}$
consists of two topological disks on $S^{2}$, and for $n=4$ the
set $\NN_{\sigma}$ consists of a single topological disk on $S^{2}$.
Using Theorem \ref{thm:MainIndexTheorem} we obtain
\[
\deg(\mult_{\sigma})=\begin{cases}
2 & \hbox{\,\,for\,\,}n=1,\\
4 & \hbox{\,\,for\,\,}n=4.
\end{cases}
\]
\begin{example}
Let us assume that there are no positive charges and that we have
$3$ negative charges located at $1/3,\exp(2\pi/3)/3,\exp(4\pi/3)/3$
. The set $\NN_{\sigma}\subset S^{2}$ has 4 holes (see Figure \ref{fig:NsigmaSphere})
so that $\chi(\NN_{\sigma})=-3$, an hence
\[
\deg(\mult_{\sigma})=12.
\]
 
\end{example}

\end{example}

\section{\label{sec:Proof-Mult-As-Kernel}Proof of theorem \ref{thm:MultiplicityAsKernelBundle}}

We will prove a slightly more general result. Let $E_{1}$ and $E_{2}$
be $\SO(2)$ vector bundles of rank $2$ over a surface $X$; for
our particular case $E_{1}=\cotangent X$ and $E_{2}=\symz F$. The
bundle of homomorphisms $\hbox{Hom}(E_{1},E_{2})$ stratifies as
\[
\hbox{Hom}(E_{1},E_{2})=\bigcup_{i=0}^{2}[\hbox{Hom}(E_{1},E_{2})]_{i}
\]
where 
\[
[\hbox{Hom}(E_{1},E_{2})]_{i}=\{A\in\hbox{Hom}(E_{1},E_{2})|\dim(\ker(A)=i\},
\]
is a smooth sub-manifold of $\hbox{Hom}(E_{1},E_{2})$ of co-dimension
$i^{2}$ (see \cite[pg. 28]{kn:arnold4}). If for a section $s:X\rightarrow\hbox{Hom}(E_{1},E_{2})$
we define
\begin{align*}
\SS_{s} & =\{x\in X|\dim(\ker(s(x))>0\},\\
\KK_{s} & =\bigcup_{x\in\SS_{s}}\ker(s(x)),
\end{align*}
we then have that 
\[
\SS_{s}=\SS_{s,1}\cup\SS_{s,2}\hbox{\,\,\,and\,\,\,}\KK_{s}=\KK_{s,1}\cup\KK_{s,2},
\]
where
\[
\SS_{s,i}=\{x\in X|\dim(\ker(s(x)))=i\}=s^{-1}\left([\hbox{Hom}(E_{1},E_{2})]_{i}\right)
\]
and
\[
\KK_{s,i}=\bigcup_{x\in\SS_{s,i}}\ker(s(x)).
\]
For sections in general position the set $\SS_{s,2}$ is empty since
$[\hbox{Hom}(E_{1},E_{2})]_{2}$ has co-dimension 4 and $X$ has dimension
2. Hence, we can generically assume that
\[
\SS_{s}=\SS_{s,1}\hbox{\,\,and\,\,\,}\KK_{s}=\KK_{s,1}.
\]
We define
\[
S(E_{1})=\{v\in E_{1}|<v,v>=1\},
\]
and for a given section $s:X\rightarrow\hbox{Hom}(E_{1},E_{2})$ let
$\sigma_{s}:E_{1}\rightarrow E_{2}$ be given by
\[
\sigma_{s}(v)=s(x)v\hbox{\,\,\,for\,\,\,}v\in E_{x}.
\]

\begin{lem}
Let $s$ be a section of $\hbox{Hom}(E_{1},E_{2})$. Under the generic
assumption that $\SS_{s}=\SS_{s,1}$, the section $s$ is transversal
to $[\hbox{Hom}(E_{1},E_{2})]_{1}$ if and only if $\sigma_{s}|S(E_{1})$
is transversal to the zero section of $E_{2}$.
\end{lem}

\begin{proof}
We can locally trivialize $E_{1}$ and $E_{2}$ by choosing local
orthonormal frames over an appropriate open set $U\subset X$. Using
these trivilalizations we can write
\[
S(E_{1}|U)=U\times S^{1}=\{(x,v_{1},v_{2})|v_{1}^{2}+v_{2}^{2}=1\}
\]
and
\[
\sigma_{s}(x,v)=s(x)v
\]
where 
\[
s(x)=\left(\begin{array}{cc}
a_{1}(x) & b_{1}(x)\\
a_{2}(x) & b_{2}(x)
\end{array}\right)\hbox{\,and\,\,\,}v=\left(\begin{array}{c}
v_{1}\\
v_{2}
\end{array}\right).
\]
Consider a curve $t\mapsto(x(t),v(t))\in U\times S^{1}$ which allows
us to express $s$ and $v$ in term of $t$, so that
\[
d\sigma_{s}(x,v)(\dot{x},\dot{v})=\dot{s}v+s\dot{v}.
\]
For a fixed $(x,v)$ with $v_{1}^{2}+v_{2}^{2}=1$ that satisfies
$\sigma_{s}(x,v)=0$ we must have scalars $\alpha_{1},\alpha_{2}\in\RR$
such that 
\begin{equation}
(a_{i}(x),b_{i}(x))=\alpha_{i}(-v_{2},v_{1})\label{eq:sigma0EqualsZeroEquation}
\end{equation}
The condition that $v$ is unitary implies that $\dot{v}=\beta(-v_{2},v_{1})$
for a scalar $\beta\in\RR$. Using the above identities we obtain
\[
d\sigma_{s}(x,v)(\dot{x},\dot{v})=\left(\begin{array}{c}
\dot{\gamma}_{1}\\
\dot{\gamma}_{2}
\end{array}\right)+\beta\left(\begin{array}{c}
\alpha_{1}\\
\alpha_{2}
\end{array}\right)
\]
where
\[
\dot{\gamma}_{i}=\dot{a}_{i}v_{1}+\dot{b}_{i}v_{2}.
\]
Hence, the condition of $\sigma_{s}$ being transversal to the zero
section of $\symz F$ at a given point $(x,v)$ means that there exists
a $\dot{x}\in\RR^{2}$ such that 
\begin{equation}
\det\left(\begin{array}{cc}
\dot{\gamma}_{1} & \dot{\gamma}_{2}\\
\alpha_{1} & \alpha_{2}
\end{array}\right)\not=0.\label{eq:TransvesalityEqsForSigma0}
\end{equation}
The condition that $s$ is transversal to $[\hbox{Hom}(E_{1},E_{2})]_{1}$
is equivalent to $0\in\RR$ being a regular value of $f_{s}=\det(s)=a_{1}b_{2}-a_{2}b_{1}$.
This last conditions implies that for $x\in f^{-1}(0)$ there must
exists $\dot{x}$ such that 
\[
df_{s}(x)\dot{x}=\dot{a}_{1}b_{2}+a_{1}\dot{b}_{2}-\dot{a}_{2}b_{1}-a_{2}\dot{b}_{1}\not=0
\]
If $f_{s}(x)=0$ then we must have that $\ker(s(x))>0.$ This means
that there exists $(v_{1},v_{2})$ with $v_{1}^{2}+v_{2}^{2}=1$ and
such that (\ref{eq:sigma0EqualsZeroEquation}) holds. Using this,
a simple calculation shows that
\[
df(x)\dot{x}=\det\left(\begin{array}{cc}
\dot{\gamma}_{1} & \dot{\gamma}_{2}\\
\alpha_{1} & \alpha_{2}
\end{array}\right).
\]
Hence, the condition $df(x)\dot{x}\not=0$ is the same as (\ref{eq:TransvesalityEqsForSigma0}).
\end{proof}
We conclude the proof of Theorem \ref{thm:MultiplicityAsKernelBundle}
as follows. The condition of $s$ being transversal to $[\hbox{Hom}(E_{1},E_{2})]_{1}$
holds for symbols in general position. From the above Lemma this implies
that the corresponding symbol $\sigma=\sigma_{s}$ is also transversal
to the zero section of $\symz F.$ Furthermore, the transversality
of $s$ to $[\hbox{Hom}(E_{1},E_{2})]_{1}$ implies that 
\[
\SS_{\sigma}=\sigma^{-1}\left([\hbox{Hom}(E_{1},E_{2})]_{1}\right)
\]
is a smooth sub-manifold of $X$. Finally, the bundle $\KK_{\sigma}$
is smooth since from formula \ref{eq:sigma0EqualsZeroEquation} we
have that 
\[
\KK_{\sigma}|U=\left\{ \left(x,\frac{(b_{1}(x),-a_{1}(x))}{(a_{1}^{2}(x)+b_{1}^{2}(x))^{1/2}}\right)\right\} _{x\in\SS_{\sigma}},
\]
and the functions $a_{1}$ and $b_{2}$ are smooth functions with
$(a_{1}(x),b_{1}(x))\not=(0,0)$ for $x\in S_{\sigma}|U.$

\section{\label{sec:Proof-ISOMORPHISM}Proof of theorem \ref{thm:ISOMORPHISM}}

An homomorphism from $(\RR^{2})^{*}$ to $\symz{\RR^{2}}$ can be
seen as an element $\RR^{2}\otimes_{\RR}\symz{\RR^{2}}.$ The rotation
matrix $R_{\theta}\in\SO(2)$ 
\[
R_{\theta}=\left(\begin{array}{cc}
\cos(\theta) & -\sin(\theta)\\
\sin(\theta) & \cos(\theta)
\end{array}\right)
\]
acts on elements of the form $v\otimes A\in\RR^{2}\otimes_{\RR}\symz{\RR^{2}}$
as 
\begin{equation}
R_{\theta}\cdot(v\otimes_{\RR}A)=(R_{\theta}v)\otimes_{\RR}(R_{\theta}AR_{\text{\ensuremath{\theta}}}^{T}).\label{eq:ActionSO2OnSymbols}
\end{equation}
If we let
\[
e_{1}=(1,0),e_{2}=(0,1),f_{1}=\left(\begin{array}{cc}
1 & 0\\
0 & -1
\end{array}\right),f_{2}=\left(\begin{array}{cc}
0 & 1\\
1 & 0
\end{array}\right)
\]
then on the basis formed by $e_{1}\otimes f_{1},e_{1}\otimes f_{2},e_{2}\otimes f_{1},e_{2}\otimes f_{2}$
the action $R_{\theta}$ has matrix representation
\[
M_{\theta}=\left(\begin{array}{cccc}
\cos(\theta)\cos(2\theta) & -2\cos^{2}(\theta)\sin(\theta) & -\cos(2\theta)\sin(\theta) & \sin(\theta)\sin(2\theta)\\
\cos(\theta)\sin(2\theta) & \cos(\theta)\cos(2\theta) & -2\cos(\theta)\sin^{2}(\theta) & -\cos(2\theta)\sin(\theta)\\
\cos(2\theta)\sin(\theta) & -2\cos(\theta)\sin^{2}(\theta) & \cos(\theta)\cos(2\theta) & -2\cos^{2}(\theta)\sin(\theta)\\
\sin(\theta)\sin(2\theta) & \cos(2\theta)\sin(\theta) & \cos(\theta)\sin(2\theta) & \cos(\theta)\cos(2\theta)
\end{array}\right).
\]
The matrix $M_{\theta}$ has eigenvalues $e^{-i\theta},e^{i\theta},e^{-3\theta},e^{3i\theta}$
whose corresponding eigenvectors are the columns of the matrix
\[
\left(\begin{array}{cccc}
1 & 1 & -1 & -1\\
i & -i & -i & i\\
-i & i & -i & i\\
1 & 1 & 1 & 1
\end{array}\right).
\]
Using the real and imaginary parts of the eigenvectors of $M_{\theta}$
corresponding to $e^{i\theta}$ and $e^{3i\theta}$ we construct the
orthogonal matrix
\[
E=\left(\begin{array}{cccc}
\frac{1}{\sqrt{2}} & 0 & -\frac{1}{\sqrt{2}} & 0\\
0 & \frac{1}{\sqrt{2}} & 0 & -\frac{1}{\sqrt{2}}\\
0 & -\frac{1}{\sqrt{2}} & 0 & -\frac{1}{\sqrt{2}}\\
\frac{1}{\sqrt{2}} & 0 & \frac{1}{\sqrt{2}} & 0
\end{array}\right)
\]
that satisfies
\[
E^{T}M_{\theta}E=\left(\begin{array}{cccc}
\cos(\theta) & -\sin(\theta) & 0 & 0\\
\sin(\theta) & \cos(\theta) & 0 & 0\\
0 & 0 & \cos(3\theta) & -\sin(3\theta)\\
0 & 0 & \sin(3\theta) & \cos(3\theta)
\end{array}\right).
\]
It follows that if 

\[
cg_{1}+dg_{2}+\gamma\cdot g_{3}+\delta\cdot g_{4}=R_{\theta}\cdot(ag_{1}+bg_{2}+\alpha g_{3}+\beta g_{4})
\]
for
\begin{align*}
g_{1} & =\frac{1}{\sqrt{2}}\left(e_{1}\otimes_{\RR}f_{1}+e_{2}\otimes_{\RR}f_{2}\right)\\
g_{2} & =\frac{1}{\sqrt{2}}\left(e_{1}\otimes_{\RR}f_{2}-e_{2}\otimes_{\RR}f_{1}\right)\\
g_{3} & =\frac{1}{\sqrt{2}}\left(-e_{1}\otimes_{\RR}f_{1}+e_{2}\otimes_{\RR}f_{2}\right)\\
g_{4} & =\frac{1}{\sqrt{2}}\left(-e_{1}\otimes_{\RR}f_{2}-e_{2}\otimes_{\RR}f_{1}\right)
\end{align*}
then
\begin{eqnarray*}
c+id & = & \exp(i\theta)(a+ib),\\
\gamma+i\delta & = & \exp(i3\theta)(\alpha+i\beta).
\end{eqnarray*}
We conclude that the map $\Phi:\CC\oplus(\CC\otimes_{\CC}\CC\otimes_{\CC}\CC)\rightarrow\symz{\RR^{2}\otimes\RR}^{2}$
given by
\[
\Phi((a+ib)\oplus(\alpha+i\beta)1\otimes_{\CC}1\otimes_{\CC}1)=ag_{1}+bg_{2}+\alpha g_{3}+\beta g_{4}
\]
is equivariant with respect to the $\SO(2)=\U(1)$ actions in $\CC\oplus(\CC\otimes_{\CC}\CC\otimes_{\CC}\CC)$
and $\symz{\RR^{2}\otimes\RR^{2}}$, and hence it induces an isomorphism
$\Phi:F\oplus(F\otimes_{\CC}F\otimes_{\CC}F)\rightarrow\hbox{Hom}(F^{*},\symz F))$.
If we let
\[
\sigma=\Phi((a+ib)\oplus(\alpha+i\beta)1\otimes_{\CC}1\otimes_{\CC}1)
\]
then for $\xi=(\xi_{1},\xi_{2})\in(\RR^{2})^{*}$ we get
\begin{equation}
\sigma(\xi)=\left(\begin{array}{cc}
p_{1} & q_{1}\\
q_{1} & -p_{1}
\end{array}\right)\xi_{1}+\left(\begin{array}{cc}
p_{2} & q_{2}\\
q_{2} & -p_{2}
\end{array}\right)\xi_{2},\label{eq:SymbolCoordinatesFormula}
\end{equation}
where 
\[
(p_{1},q_{1},p_{2},q_{2})=\frac{1}{\sqrt{2}}(a-\alpha,b-\beta,-b-\beta,a+\alpha).
\]
If we use the above formulas we obtain
\[
(\sigma^{*}g)(\xi,\eta)=(\eta_{1},\eta_{2})G_{\sigma}\left(\begin{array}{c}
\xi_{1}\\
\xi_{2}
\end{array}\right)
\]
where $\xi=(\xi_{1},\xi_{2}),\eta=(\eta_{1},\eta_{2})$ and 
\[
G_{\sigma}=\left(\begin{array}{cc}
\frac{1}{2}\left((a-\alpha)^{2}+(b-\beta)^{2}\right) & b\alpha-a\beta\\
b\alpha-a\beta & \frac{1}{2}\left((a+\alpha)^{2}+(b+\beta)^{2}\right)
\end{array}\right).
\]
The above matrix has eigenvalues
\begin{align*}
\kappa_{1} & =\frac{1}{2}(|a+ib|-|\alpha+i\beta|)^{2},\\
\kappa_{2} & =\frac{1}{2}(|a+ib|+|\alpha+i\beta|)^{2},
\end{align*}
which gives the formulas for the eigenvalues of $G_{\sigma}$ stated
in the theorem. If we use polar coordinates so that 
\[
a+ib=r\exp(i\theta)\hbox{\,\,\,and\,\,\,}\alpha+i\beta=\text{\ensuremath{\rho\exp(i\varphi)}}
\]
 then it is easy to see that the eigen-spaces $L_{1}$ and $L_{2}$
corresponding to $\kappa_{1}$ and $\kappa_{2}$ are spanned by the
orthonormal vectors 
\begin{align*}
z_{1} & =\exp\left(i\left(\frac{\varphi-\theta}{2}\right)\right),\\
z_{2} & =i\exp\left(i\left(\frac{\varphi-\theta}{2}\right)\right).
\end{align*}
Hence
\begin{align*}
z_{1}^{2} & =\left(\frac{r}{\rho}\right)\left(\frac{\alpha+i\beta}{a+ib}\right),\\
z_{2}^{2} & =-\left(\frac{r}{\rho}\right)\left(\frac{\alpha+i\beta}{a+ib}\right),
\end{align*}
and
\begin{align*}
L_{1} & =\RR\cdot\left\{ z^{\flat}\in F_{x}^{*}\left|z^{2}=\left(\frac{\alpha+i\beta}{a+ib}\right)\right.\right\} ,\\
L_{2} & =\RR\cdot\left\{ z^{\flat}\in F_{x}^{*}\left|z^{2}=-\left(\frac{\alpha+i\beta}{a+ib}\right)\right.\right\} .
\end{align*}

\bibliographystyle{plain}
\bibliography{myBib}

\end{document}